\documentclass[11pt,a4paper,USenglish,thm-restate]{article}

\usepackage[utf8]{inputenc}

\usepackage[margin=1in]{geometry}

\usepackage{amssymb}
\usepackage{amsmath}
\usepackage{amsfonts}
\usepackage{amsthm}
\usepackage{nccmath}
\usepackage[b]{esvect}
\usepackage{mathrsfs}
\usepackage{mathtools}
\usepackage{centernot}
\usepackage{xspace}
\usepackage{stmaryrd}
\usepackage[noadjust]{cite} 

\usepackage{thm-restate}
\usepackage{totcount}
\usepackage[dvipsnames]{xcolor}

\usepackage[colorinlistoftodos]{todonotes}
\usepackage{todonotes}
\usepackage{verbatim}

\usepackage{graphicx}

\usepackage{authblk}
\usepackage{lscape}
\usepackage{float}
\usepackage{supertabular}
\usepackage{multirow}
\usepackage{longtable}
\usepackage{enumerate}
\usepackage{setspace}
\usepackage[inline]{enumitem}
\usepackage[framemethod=tikz]{mdframed} 
\usepackage{algpseudocode}

\usepackage{hyperref} 
\usepackage{cleveref}
\usepackage{nameref}
\usepackage{url}

\usepackage{makecell}

\usepackage{caption}
\usepackage{subcaption}

\usepackage{textcomp}
\usepackage{pgfplots}
\usepgfplotslibrary{fillbetween}

\usepackage{graphicx,wrapfig,lipsum}

\newtheorem{theorem}{Theorem}
\newtheorem{lemma}[theorem]{Lemma}
\newtheorem{corollary}[theorem]{Corollary}
\newtheorem{definition}[theorem]{Definition}

\makeatletter
\DeclareRobustCommand*{\bfseries}{%
  \not@math@alphabet\bfseries\mathbf
  \fontseries\bfdefault\selectfont
  \boldmath
}
\makeatother

\makeatletter
\renewcommand{\paragraph}[1]{\medskip\noindent{\bf #1}}
\makeatother

\RequirePackage{totcount}
\RequirePackage{color}
\RequirePackage[colorinlistoftodos]{todonotes}

\newtotcounter{notecount}
\newcommand{\notewarning}{%
\ifnum\totvalue{notecount}>0%
 \vspace{1ex}
\begin{center}
 \begin{tikzpicture}[baseline=(A.south)]
    \node (A) [] at (0,0){};
    \node [rounded corners=1pt,rectangle, draw=red, fill=red!20,text=black](B) at (0.1ex,0ex){
        \Large \raggedright {\bf Warning:} There are still some notes left!
    };
 \end{tikzpicture}
\end{center}
 \vspace{1ex}
\fi
}
\makeatletter
\def\myaddcontentsline#1#2#3{%
  \addtocontents{#1}{\protect\contentsline{#2}{#3}{Section \thesubsection\ at p. \thepage}{}}}
\renewcommand{\@todonotes@addElementToListOfTodos}{%
    \if@todonotes@colorinlistoftodos%
        \myaddcontentsline{tdo}{todo}{{%
            \colorbox{\@todonotes@currentbackgroundcolor}%
                {\textcolor{\@todonotes@currentbackgroundcolor}{o}}%
            \ \@todonotes@caption}}%
    \else%
        \myaddcontentsline{tdo}{todo}{{\@todonotes@caption}}%
   \fi}%
\newcommand*\mylistoftodos{%
  \begingroup
       \setbox\@tempboxa\hbox{Section 9.9 at p. 99}%
       \renewcommand*\@tocrmarg{\the\wd\@tempboxa}%
       \renewcommand*\@pnumwidth{\the\wd\@tempboxa}%
       \listoftodos%
  \endgroup
}
\makeatother
\definecolor{lightgreen}{rgb}{0.86, 0.93, 0.78}
\definecolor{bordergreen}{rgb}{0.55, 0.76, 0.74}
\definecolor{lightblue}{rgb}{0.70, 0.90, 0.99}
\definecolor{borderblue}{rgb}{0.01, 0.66, 0.96}
\definecolor{lightamber}{rgb}{1, 0.93, 0.70}
\definecolor{borderamber}{rgb}{1, 0.76, 0.03}
\definecolor{lightcolor4}{rgb}{ 0.93, 0.70, 1}
\definecolor{bordercolor4}{rgb}{0.76, 0.03, 1}
\definecolor{lightcolor5}{rgb}{0.78,0.86,0.93}
\definecolor{bordercolor5}{rgb}{0.74,0.55,0.76}

\RequirePackage{algpseudocode}
\newcommand{\algoHead}[1]{\vspace{0.2em} \underline{\textbf{#1}} \vspace{0.3em}}

\makeatletter
\algnewcommand{\ExtendedState}[1]{\State
\parbox[t]{\dimexpr\linewidth-\ALG@thistlm}{\hangindent=\algorithmicindent\strut\hangafter=3#1\strut}}
\makeatother
\algnewcommand\algorithmicinput{\textbf{Input:}}
\algnewcommand\Input{\item[\algorithmicinput]}

\algrenewcommand{\algorithmiccomment}[1]{{\color{gray}// #1}}
\algnewcommand{\IIf}[1]{\State\algorithmicif\ #1\ \algorithmicthen}
\algnewcommand{\EndIIf}{\unskip\ \algorithmicend\ \algorithmicif}

\algdef{SE}
[UPON]
{Upon}
{EndUpon}
[1]
{Upon {#1}} 
{\textbf{end upon}} 
\algtext*{EndUpon} 

\algdef{SE}
[WHEN]
{When}
{EndWhen}
[1]
{When {#1}} 
{\textbf{end when}} 
\algtext*{EndWhen} 

 \RequirePackage{varwidth}
 \RequirePackage{color}
 \RequirePackage[most]{tcolorbox}

 \newtcolorbox{titlebox}[5]{enhanced,center,colframe=black,colback=white,boxrule={#3},arc={#2},auto outer arc,%
 breakable,pad at break*=5pt,vfill before first,before={
 },after={\par\smallskip},top=12pt,left=4pt,%
 enlarge top by=2pt,
 fontupper=\small,
 title={\rule[-.3\baselineskip]{0pt}{\baselineskip}\normalsize\sffamily\bfseries #1}, varwidth boxed title*=-30pt, 
 attach boxed title to top left={yshift=-10pt,xshift=10pt}, coltitle=black,
 boxed title style={colback=white,boxrule={#5},arc={#4},auto outer arc}
 }

 \newenvironment{dianabox}[1]
 {\begin{titlebox}{\normalfont #1}{0.5pt}{0.5pt}{1pt}{0.75pt}}
 {\end{titlebox}}

\makeatletter
\let\orgdescriptionlabel\descriptionlabel
\renewcommand*{\descriptionlabel}[1]{%
  \let\orglabel\label
  \let\label\@gobble
  \phantomsection
  \edef\@currentlabel{#1}%
  \let\label\orglabel
  \orgdescriptionlabel{#1}%
}
\makeatother

\let\emptyset\varnothing

\DeclarePairedDelimiter\abs{\big\lvert}{\big\rvert}

\newcommand{\party}{P\xspace}
\newcommand{\PS}{\mathcal{P}} 

\newcommand{\acs}{\textsc{ACS}}

\newcommand{\val}[0]{\textsc{val}}

\newcommand{\inputconfigs}[0]{\mathcal{I}}

\newcommand{\inputt}[0]{\textsc{in}}
\newcommand{\outputt}[0]{\textsc{out}}
\newcommand{\ba}[0]{\textsc{BA}}

\newcommand{\validity}[0]{\textsc{val}}
\newcommand{\parties}[0]{\textsc{parties}}

\newcommand{\execution}{\varepsilon}

\newcommand{\neighbors}[0]{\textsc{neighbors}}
\newcommand{\similar}[0]{\textsc{similar}}

\title{Validity in Network-Agnostic Byzantine Agreement}

\author[1]{Andrei Constantinescu}
\author[1]{Marc Dufay}
\author[2]{Diana Ghinea}
\author[1]{Roger Wattenhofer}

\affil[1]{ETH Z{\"u}rich\\ \{aconstantine, mdufay, wattenhofer\}@ethz.ch}
\affil[2]{Lucerne University of Applied Sciences and Arts\\ diana.ghinea@hslu.ch}

\date{\vspace{-5ex}}

\begin{document}

\maketitle

\begin{abstract}
Byzantine Agreement (BA) considers a setting of $n$ parties, out of which up to $t$ can exhibit byzantine (malicious) behavior. Honest parties must decide on a common value (agreement), which must belong to a set determined by the honest inputs (validity). Depending on the use case, this set can grow or shrink, leading to various possible desiderata collectively known as validity conditions.
Varying the validity property requirement can affect the regime under which BA is solvable.

Our work investigates how the selected validity property impacts BA solvability in the network-agnostic model, where the network can either be synchronous with up to $t_s$ byzantine parties or asynchronous with up to $t_a \leq t_s$ byzantine parties. We give necessary and sufficient conditions for a validity property to render BA solvable, both for the case with cryptographic setup and for the one without. This traces the precise boundary of solvability in the network-agnostic model for every validity property. Our proof of sufficiency provides a universal protocol, that achieves BA for a given validity property whenever the provided conditions are satisfied. 

We note that, for any non-trivial validity property, the condition $2 \cdot t_s + t_a < n$ is necessary for BA to be solvable, even with cryptographic setup. Specializing this claim to $t_a = 0$ gives that $t < n / 2$ is required whenever one expects a purely synchronous protocol to also work in an asynchronous network when there are no corruptions. This is especially surprising given that, for some validity properties, $t < n$ is a sufficient condition without the last stipulation.
\end{abstract}



\newpage
\setcounter{tocdepth}{2} 
\tableofcontents

\thispagestyle{empty} 

\newpage
\pagenumbering{arabic}

\section{Introduction}

Achieving agreement among the parties involved in a distributed system is crucial for maintaining consistent views. This becomes particularly challenging due to the potential for parties' failures, which can range from benign crashes to malicious (byzantine) behavior. Byzantine Agreement ($\ba$) is an extensively studied problem in distributed computing that tackles this challenge. It seeks to establish a common value amongst a set of $n$ parties even when up to $t$ parties exhibit byzantine behavior.
A crucial aspect of $\ba$ lies in its validity condition, which requires that the value agreed upon reflects the honest parties' proposals rather than being a default or arbitrary value. The traditional definition of $\ba$ considers the so-called \emph{strong unanimity} (also known as \emph{strong validity}): if all honest parties propose the same value $v$, the agreed-upon output must be $v$.
This is a powerful guarantee in applications concerning binary decisions, as it coincides with \emph{honest-input validity}, which requires agreement on an honest input. However, strong unanimity fails to provide meaningful outputs for larger input spaces. For instance, consider a set of parties running a $\ba$ protocol to agree on a room's temperature.
Minor measurement errors are inherent, and hence strong unanimity allows agreement on a temperature proposed by corrupted parties. Similar challenges arise in scenarios such as deciding on a location using GPS coordinates \cite{BOUZID20103154} or when the map is modeled as a graph \cite{DISC:NoRy19,constantinescu_ghinea_convex_consensus}.

While achieving honest-input validity in scenarios where the input space is large (size $\omega(n)$) is impossible \cite{Nei94}, the literature offers a plethora of weaker alternatives that are stronger and more meaningful than strong unanimity. For instance, one may avoid corrupted outputs by enhancing strong unanimity with an additional condition called \emph{intrusion tolerance} \cite{EUROCRYPT:Mose24,PODC24:GhLiWa,PODC25:GhLiWa}: 
the honest parties either agree on an honest input or on a special symbol $\bot$.
In the aforementioned scenarios of deciding on a measurement or a meeting point, a highly suitable alternative is the so-called \emph{convex validity}: the output agreed upon must be in the honest parties' inputs convex hull, i.e., within the range of honest inputs if the input space is a subset of $\mathbb{R}$ \cite{PODC:VaiGar13,DIST:MHVG15, DISC:NoRy19,constantinescu_ghinea_convex_consensus,PODC24:GhLiWa,PODC25:GhLiWa}. Stronger variants for real values require the outputs to be close to the honest inputs' median \cite{OPODIS:StolWat15, OPODIS:CGHWW23}, or to the $k$-th lowest honest input \cite{IEEE:MelWat18}. However, the honest-range approach is not a universal solution: this would carry no meaning in a voting problem if we represented candidates with integers. Instead, approaches based on social choice theory (\emph{Pareto validity}) lead to more appropriate validity definitions \cite{WINE18:ByzantineVoting}.

The multitude of validity definitions leads to a natural question: what are the necessary and sufficient conditions for achieving $\ba$ with a given validity property, i.e., \emph{to solve a given validity property}? This question can be concerned with multiple aspects, such as resilience thresholds $t$, round complexity, or message complexity. In addition, the conditions may depend on whether a cryptographic setup and randomization are available. The communication model assumed also plays an important role: one extreme is the synchronous model, where all parties have synchronized clocks, and all messages get delivered within a predefined amount of time $\Delta$. This enables elegant protocols that operate in rounds, but that may fail in a real-life network where sporadic issues are possible.
The other extreme is the asynchronous model, which only assumes that messages get delivered eventually. The asynchronous model comes with highly robust protocols, but also with important drawbacks:
\begin{enumerate*}[label=(\alph*)]
    \item lower resilience threshold: e.g., $t < n/3$ as opposed to $t < n/2$ (assuming digital signatures) for strong unanimity, or $t < n/4$ as opposed to $t < n/3$ for convex validity in $\mathbb{R}^2$;
    \item randomization is a requirement
    \cite{FLP}.
\end{enumerate*}
The middle ground between the two models has also been considered. The partially synchronous model \cite{JACM:DLS88} bridges the gap between the two extremes by assuming that the network is initially asynchronous and eventually becomes synchronous. In the synchronous and partially synchronous models, the solvability of validity conditions using deterministic protocols is completely understood \cite{PODC24:CGGKPV, PODC23:CGGKV}.

In this paper, we will be concerned with a different paradigm for bridging the gap between synchrony and asynchrony, namely the \emph{network-agnostic} model,
which has attracted increased attention in recent years \cite{TCC:BKL19, Crypto:BLL20, TCC:DHL21, CCS:AtsRen21, PODC:ApChCh22, PODC:GhLiWa22, SPAA:GhLiWa23, constantinescu_ghinea_convex_consensus}. In this model, parties are initially unaware of whether the network is synchronous or not: if it is synchronous, at most $t_s$ of the parties involved may be corrupted, and otherwise only $t_a \leq t_s$ of the parties may be corrupted. Network-agnostic protocols are designed to provide guarantees in both cases. Hence, we ask the following question:

\vspace{0.2cm}
{
    \centering
    \emph{What are the necessary and sufficient conditions for achieving\\ network-agnostic $\ba$ with a given validity property?}\par
}
\vspace{0.2cm}



\subsection{Our Contribution}
We provide a complete characterization for achieving network-agnostic $\ba$ with a given validity property, establishing tight necessary and sufficient conditions.

We first show that, even if cryptographic setup is available, the condition $n > 2 \cdot t_s + t_a$ is a requirement for any non-trivial validity condition to be solvable (i.e., a condition for which simply outputting a default value does not suffice). When no cryptographic setup is available, we show the stronger requirement of $n > 3 \cdot t_s$.
Our proof for the latter, in fact, works in the synchronous model and, therefore, strengthens the characterization provided by \cite{PODC24:CGGKPV} for the synchronous model. In particular, \cite{PODC24:CGGKPV} only focuses on deterministic protocols, while our proofs rely on different techniques that also apply to randomized protocols.

Afterwards, regardless of whether cryptographic setup is available or not, we add one more necessary condition, which is an adaptation of
the \emph{similarity condition} of \cite{PODC23:CGGKV} and the \emph{containment condition} of \cite{PODC24:CGGKPV} to the network-agnostic model. 
Roughly, this requires that \emph{similar} honest inputs' configurations have a valid output in common.
The term \emph{similar} captures that some of the proposed inputs may come from corrupted parties, and that, in asynchronous networks, some honest inputs may be missing due to high network delays.

Finally, we show that, together, the aforementioned conditions are also sufficient by providing a universal protocol, that achieves network-agnostic $\ba$ for a given validity property whenever these conditions are satisfied. This is a more general variant of the protocol of \cite{constantinescu_ghinea_convex_consensus}. In particular, the requirement for solvability is precisely $n > 2 \cdot t_s + t_a$ together with the similarity condition assuming cryptographic setup, and $n > 3 \cdot t_s$ together with the similarity condition assuming no cryptographic setup. Our protocol is randomized, which is a requirement when the network may be asynchronous and $t_a > 0$ \cite{FLP}.




\subsection{Related Work}
\paragraph{General validity conditions.} 
The foundational investigation into general validity properties was initiated by Civit et al. \cite{PODC23:CGGKV} for the partially synchronous model. Subsequently, Civit et al. \cite{PODC24:CGGKPV} embarked on a follow-up study, extending their analysis to the synchronous model. Both works provide a complete characterization, identifying the necessary and sufficient conditions for solving a validity property deterministically. 
We also note that the contributions of \cite{PODC23:CGGKV, PODC24:CGGKPV} extend beyond this characterization, an important side of these works lying in the exploration of lower bounds on message complexity. This generalizes the well-established Dolev-Reischuk bound on message complexity for $\ba$ with strong unanimity \cite{JACM:DoRe85} to encompass the broader landscape of non-trivial validity properties. Considerations of message complexity are outside our scope.
We also note that \emph{probabilistic} validity properties are outside the scope of the analysis of Civit et al, and also outside our scope. A notable example is \emph{qualitative validity}, introduced by Goren et al \cite{GoMoSp23}.


\paragraph{Network-Agnostic $\ba$ and particular validity conditions.} 
Designing protocols that achieve security guarantees in both synchronous and asynchronous networks has been the subject of an extensive line of work. The network-agnostic paradigm was introduced by Blum, Katz and Loss \cite{TCC:BKL19}. The work of \cite{TCC:BKL19} shows that, if a public key infrastructure is provided, $\ba$ with strong unanimity can be achieved if and only if  $n > 2 \cdot t_s + t_a$. Further works on network-agnostic $\ba$ with strong unanimity have focused on improving the efficiency guarantees  \cite{TCC:DHL21, EUROCRYPT:Mose24}.

Due to its broad applicability, convex validity within the network-agnostic communication paradigm has attracted increased attention. Ghinea, Liu-Zhang and Wattenhofer \cite{PODC:GhLiWa22,SPAA:GhLiWa23} have investigated the feasibility of achieving convex validity for a weaker variant of $\ba$, known as Approximate Agreement \cite{JACM:DLPSW86, OPODIS:AAD04}. In particular, \cite{PODC:GhLiWa22} shows that Approximate Agreement on real numbers is solvable under the same necessary and sufficient condition $n > 2 \cdot t_s + t_a$ assuming a public key infrastructure. Building on the previous, \cite{SPAA:GhLiWa23} gives sufficient conditions for the multidimensional variant of the problem that match the known requirements in the pure synchronous and asynchronous models \cite{PODC:VaiGar13,DIST:MHVG15}. Returning to the non-approximate version, Constantinescu et al. \cite{constantinescu_ghinea_convex_consensus}
have provided the tight conditions for  network-agnostic $\ba$ with convex validity for abstract convex spaces.
In this case, the conditions include $n > 2 \cdot t_s + t_a$ or, if no cryptographic setup is available, $n > 3 \cdot t_s$, along with a few additional conditions that depend on the Helly number of the convexity space.

Other takes on network-agnostic $\ba$ have been considered, such as \emph{scaling} the validity guarantees with the network conditions: for real-valued inputs, \cite{OPODIS:CGHWW23} proposes a protocol for $\ba$ with median validity guaranteeing that the output is closer to the honest inputs' median when the network is synchronous than when it is asynchronous. This protocol simultaneously matches the optimal closeness guarantees for purely synchronous \cite{OPODIS:StolWat15,IEEE:MelWat18} and purely asynchronous \cite{OPODIS:CGHWW23} networks. Such validity properties that scale with the network conditions are outside our scope.


\paragraph{Comparison to previous works.}
As outlined above, the conditions $n > 2 \cdot t_s + t_a$ and, if no cryptographic setup is available, $n > 3 \cdot t_s$, have been proven to be necessary for strong unanimity properties, i.e., stronger than strong unanimity \cite{TCC:BKL19}. 
Our work shows that this is a requirement for weaker validity properties as well, i.e., for any non-trivial property. 
We find this result surprising especially for \emph{weak validity}: if \emph{all parties are honest} and hold input $v$, then the output agreed upon must be $v$. Assuming a public key infrastructure, 
this property is solvable in the synchronous model for $t_s < n$, as a straightforward application of the Dolev-Strong broadcast protocol \cite{DolStr83}. On the other hand, our result implies that if we expect a $\ba$ protocol with weak validity to remain secure in the asynchronous model even for no corruptions (i.e., $t_a = 0$), then the synchronous resilience threshold steps down from $t_s < n$ to $t_s < n  / 2$. 

In contrast to the work of \cite{constantinescu_ghinea_convex_consensus} regarding network-agnostic $\ba$ with convex validity, our results move the difficulty of proving such feasibility results as a whole to only verifying whether a validity condition satisfies our similarity condition. That is, one can show that convex validity satisfies this similarity condition if and only if the necessary Helly number-based conditions of \cite{constantinescu_ghinea_convex_consensus} hold. Our impossibility arguments diverge: we investigate these under \emph{any} validity property, while the work of \cite{constantinescu_ghinea_convex_consensus} considers a fixed validity property but also shows impossibility under weaker agreement requirements. On the other hand, our protocol matching our lower bounds is a more general variant of the protocol of \cite{constantinescu_ghinea_convex_consensus}.

Our paper provides a characterization similar to those of \cite{PODC23:CGGKV, PODC24:CGGKPV} for the synchronous and partially synchronous settings, respectively. The key difference is that these two models allow for deterministic protocols, while the network-agnostic model inherently requires randomization for achieving $\ba$ when $t_a > 0$ \cite{FLP}.
Consequently, the focus shifts towards randomized protocols, requiring our proofs to employ different techniques.
While the arguments behind the \emph{containment} condition of \cite{PODC24:CGGKPV} can be easily adapted for randomized protocols, this is not immediate for their proof that $n > 3 \cdot t_s$ is necessary when no cryptographic setup is available. Our proof for this lower bound, in fact, assumes the synchronous setting and, therefore, strengthens the characterization of \cite{PODC24:CGGKPV}. Summing up, their necessary conditions now hold even for randomized protocols, and, as shown in their paper, they can be matched by deterministic protocols. 

\section{Preliminaries}\label{sec:preliminaries}

We consider a setting with $n$ parties $\PS = \{\party_1, \party_2, \ldots, \party_n\}$ running a protocol in a fully-connected network, where links model authenticated channels.
We will be working in the \emph{network-agnostic} model: the network may be synchronous, or asynchronous, and the parties are not aware a priori of the type of network. If the network is synchronous, the parties hold perfectly synchronized clocks and each message is delivered within a publicly known amount of time $\Delta$. Otherwise, if the network is asynchronous, messages can get delayed for an arbitrary amount of time, and clocks may not be synchronized.

\paragraph{Adversary.} We assume a central adversary that may corrupt up to $t_s$ of the $n$ parties if the network is synchronous, and up to $t_a \leq t_s$ parties if the network is asynchronous. Corrupted parties permanently become byzantine, and hence may deviate arbitrarily (maliciously) from the protocol. 
The adversary additionally controls the message delivery schedule, subject to the conditions of the network type.
Our impossibility results assume a static adversary (i.e., chooses which parties to corrupt at the beginning of the protocol's execution). Our protocol, on the other hand, provides security even against an adaptive adversary (i.e., chooses which parties to corrupt at any point in the protocol's execution).

\paragraph{Cryptographic setup.} We will consider both settings with and without cryptographic setup. Protocols assuming a cryptographic setup will make use of a public key infrastructure (PKI) and a secure signature scheme, and only hold against a computationally bounded adversary. For simplicity of presentation, we assume that the signatures are perfectly unforgeable. 

\paragraph{Byzantine Agreement and Validity.} 
A Byzantine Agreement ($\ba$) protocol assumes that each (honest) party holds a value $v_{\inputt} \in V_{\inputt}$ as input, and enables the parties to agree on a common output $v_{\outputt} \in V_{\outputt}$ satisfying a given validity condition. We assume that $V_{\inputt}$ is at most countably infinite.\footnote{We will need to take an intersection of events that happen almost surely over the set of possible input configurations. If $V_{\inputt}$ is uncountably infinite, the resulting event does not necessarily happen almost surely, breaking some of our arguments.}
We discuss this in \cref{section:uncountable-stuff}. 

In the following, we present each property that a $\ba$ protocol needs to satisfy.
The first property is \emph{termination}, which may be deterministic (for synchronous protocols) or probabilistic:

\vspace{0.15cm}
\begin{itemize}[nosep]
\item \textbf{(Termination)} Every honest party decides on an output $v_{\outputt}$.
\item \textbf{(Probabilistic Termination)} As time goes to infinity, the probability that an honest party has not yet decided on an output value $v_{\outputt}$ tends to $0$.
\end{itemize}
\vspace{0.15cm}

The second property that $\ba$ requires is \emph{agreement}, as defined below.

\vspace{0.15cm}
\begin{itemize}[nosep]
\item \textbf{(Agreement)} If two honest parties output $v_{\outputt}$ and $v_{\outputt}'$, then $v_{\outputt} = v_{\outputt}'$.
\end{itemize}
\vspace{0.15cm}
Before describing the validity property, we need to define \emph{input configurations}.
Regardless of the nature of the adversary, these are defined by looking at the honest parties' inputs \emph{after} the adversary has decided which parties to corrupt. Hence, in impossibility results, where we consider a static adversary, input configurations are defined before the protocol's execution.
An \emph{input configuration} is a set $I \subseteq \PS \times V_{\inputt}$ consisting of pairs of honest parties with their inputs: if $(v, \party) \in I$, then $\party$ is an honest party with input $v$. Naturally, no party occurs twice in $I$ (i.e., honest parties cannot simultaneously have two inputs). We use the notation $\parties(I)$ to refer to the set of (honest) parties in the input configuration $I$. Note that, if $\party \notin \parties(I)$, then $\party$ is corrupted in the input configuration $I$.
Let $\inputconfigs = \{\text{input configurations }I \subseteq \PS \times V_{\inputt}\text{ such that }\abs{I} \geq n - t_s \}$ denote the set of all possible input configurations.
We also note the inclusion relation for input configurations: for $I, J \in \inputconfigs$, $J \subseteq I$ if and only if $\parties(J) \subseteq \parties(I)$ and the parties in $\parties(J)$ have the same input value in both $I$ and $J$. 
We say that an input configuration $I$ is maximal if $\parties(I) = \PS$. Moreover, as $V_{\inputt}$ is at most countably infinite, the size of $\inputconfigs$ is at most countably infinite.
A validity property is then defined by a mapping $\validity : \inputconfigs \to 2^{V_\outputt}$ from honest parties' inputs
to \emph{valid} outputs:

\vspace{0.15cm}
\begin{itemize}[nosep]
\item \textbf{(Validity)} If $I \in \inputconfigs$ is the input configuration defined by the honest parties and their inputs, then no honest party outputs $v_{\outputt} \notin \validity(I)$.
\end{itemize}
\vspace{0.15cm}

This validity definition matches the one used in \cite{civit2023byzantine, PODC24:CGGKPV}. We say that a validity property $\val$ is \emph{trivial} if
$
    \bigcap_{I \in \inputconfigs} \validity(I) \neq \emptyset.
$
Note that, if this condition holds, we can achieve validity, agreement, and termination with no communication: parties output a value in $ \bigcap_{I \in \inputconfigs} \validity(I)$.

A validity property $\validity$ is \emph{solvable} if there is a $\ba$ protocol \emph{solving} $\validity$, as defined below.

\begin{definition}
    A protocol $\Pi$ is a $(t_s, t_a)$-secure $\ba$ protocol solving a validity property $\val$ if it achieves probabilistic termination, agreement, and validity for the given property $\val$ even when up to $t_s$ parties are corrupted if it runs in a synchronous network, and even when up to $t_a$ parties are corrupted if it runs in an asynchronous network.
\end{definition}

\begin{definition}
    A protocol $\Pi$ is a $t_s$-secure $\ba$ protocol solving a validity property $\val$ if, when running in a synchronous network, it achieves (probabilistic) termination, agreement, and validity for the given property $\val$ even when up to $t_s$ parties are corrupted.
\end{definition}

One might be tempted to believe that a $(t_s, 0)$-secure protocol is simply a $t_s$-secure protocol, but the difference is rather subtle. Namely, a $(t_s, 0)$-secure \emph{network-agnostic} $\ba$ protocol also provides guarantees in an asynchronous network if all parties are honest. On the other hand, a $t_s$-secure (synchronous) $\ba$ protocol is not required to provide any guarantees if the synchrony assumptions fail, even if there are no corruptions. We add that this subtle difference would not apply to a purely asynchronous variant of the definition, which is equivalent to $(t_a, t_a)$-secure \emph{network-agnostic} $\ba$ protocol.

\paragraph{Randomness.} Our work covers $\ba$ protocols which can run in the asynchronous setting. As a result of FLP \cite{FLP}, the protocol considered must be randomized. We consider the randomness as a black box where, at each instant, a party can ask for one or multiple random bits, each set to $0$ or $1$ uniformly and independently. So the randomness from a party's point of view can be seen as an infinite bitstring being progressively read. There may also be shared randomness, which gives the same result to each party. Therefore, when running a protocol, we can consider its behavior over a probabilistic space $(\Omega, \mathcal{F}, \mu)$ where $\Omega = \left( \{0,1\}^\mathbb{N}\right)^k$ for some $k > 0$ is the set of all possible random bits parties will read. $\mathcal{F}$ is the $\sigma$-algebra generated by taking all possible prefixes from each bitstring, and $\mu$ is the resulting probability measure from having each bit following independently a Bernoulli random variable $\mathcal{B}(0.5)$. From this point on, when mentioning probabilities, like almost surely properties, we refer to the probabilistic space given above.

\paragraph{Executions.} For a protocol $\Pi$, we define an \emph{execution} $\execution$ to be a particular feasible run of $\Pi$.
In particular, $\execution$ contains the input configuration $I \in \inputconfigs$ from which the protocol started, the behavior of the byzantine parties, and the scheduler's behavior (including whether the network was synchronous or asynchronous).
Because an execution depends on the randomness, it is actually a random variable $\execution(\omega)$. We then say that an execution \emph{decides} if all honest parties in the execution eventually decide an output. We note that, given a protocol satisfying probabilistic termination, including the scheduler and strategy of the adversary, an execution for this protocol decides almost surely.
We will also be using the term \emph{canonical} to refer to executions occurring in a synchronous network
where all messages are delivered exactly $\Delta$ units of time after being sent and where all corrupted parties crash right at the beginning of the protocol (i.e., they do not send any messages). We add that, for any given input configuration, there is a unique canonical execution (which, recall, is a random variable).

For a given protocol $\Pi$ and randomness $\omega \in \Omega$, we say that two executions $\execution_1(\omega)$ and $\execution_2(\omega)$ \emph{cannot be distinguished} by a party $\party$ that is honest in both executions if it has the same initial state in $\execution_1(\omega)$ and $\execution_2(\omega)$ (i.e., input and randomness), and receives in both $\execution_1(\omega)$ and $\execution_2(\omega)$ the same messages at the same times. As a consequence, if $\party$ cannot distinguish between $\execution_1(\omega)$ and $\execution_2(\omega),$ then $\party$ is in exactly the same state at any time $T$ in $\execution_1(\omega)$ and $\execution_2(\omega)$. Hence, if $\party$ decides a value $v_{\outputt}$ at time $T$ in one of the two executions, then it also decides $v_{\outputt}$ at time $T$ in the other. We also say that an execution $\execution(\omega)$ is deciding if all honest parties eventually decide when running with randomness $\omega$. If a protocol satisfies probabilistic termination, then an execution is almost surely deciding.


\paragraph{Validity in indistinguishable executions.} Goren et al.~\cite{GoMoSp23} offer a framework to formalize indistinguishability for randomized protocols. We instead decided to take a different approach by looking at deterministic indistinguishability after fixing the randomness $\omega \in \Omega$. This can allow for easier results as we do not have to consider the whole probabilistic space at a time. For example, some of our proofs require us to define multiple executions depending on $\omega$ and cannot be achieved with the framework above. The drawback is that all the assumptions made must hold \emph{almost surely} for the proofs to be correct
. Indistinguishability is a powerful tool, especially when considering scenarios where byzantine parties follow the protocol correctly, but with inputs of their own choice.  This leads us to the following lemma.
The proof is included in Appendix \ref{appendix:preliminaries}.


\begin{restatable}{lemma}{ValidityInclusion}\label{lemma:validity-inclusion}
    Let $\validity$ be a validity property and $\Pi$ be a $(t_s, t_a)$-resilient $\ba$ protocol solving $\validity$. 
    Consider two input configurations $I, J$ such that $J \subseteq I$. 
    Then, the value agreed upon in any execution of $\Pi$ which decides 
    and where the input configuration is $I$ must be in $\validity(J)$.
\end{restatable}
Intuitively, Lemma \ref{lemma:validity-inclusion} ensures that honest parties cannot distinguish between a scenario where all parties in $\parties(I)$ are honest, and one where the parties in $\parties(I) \setminus \parties(J)$ are, in fact, byzantine.
%
With the lemma in mind, for any validity property $\validity$, we can define a new validity property $\validity'$ such that $\validity'(I) = \bigcap_{J \subseteq I} \validity(J)$ for all $I \in \inputconfigs$. Property $\validity'$ is simultaneously a stronger version of $\validity$, in that for all $I \in \inputconfigs$ we have $\validity'(I) \subseteq \validity(I)$, but it also has the property of being monotonically \emph{decreasing}, in that for $J \subseteq I$ we have $\validity'(I) \subseteq \validity'(J)$. Armed as such, the previous lemma has the following immediate corollaries:%
\begin{corollary} A protocol solves $\validity$ if and only if it solves $\validity'$. Hence, $\validity$ is solvable if and only if $\validity'$ is solvable.
\end{corollary}
\begin{corollary} \label{coro:party-max}
    A solvable validity property $\validity$ is trivial if and only if it permits deciding the same value for all maximal input configurations, i.e., $\bigcap_{I \in \inputconfigs, \parties(I) = \PS} \validity(I) \ne \emptyset$.
\end{corollary}
We end by stating a technical lemma that will be of use in the proofs presented in the subsequent sections. The proof of \cref{lemma:technical-end} is included in Appendix \ref{appendix:preliminaries}.
\begin{restatable}{lemma}{TechnicalEnd}\label{lemma:technical-end}
    Let $\validity$ be a solvable validity property and $\Pi$ a protocol solving $\validity$.
    Let $I_1 \in \inputconfigs$ be a maximal input configuration. If, for every maximal input configuration $I_2 \in \inputconfigs$, the canonical executions of $\Pi$ for $I_1$ and $I_2$ decide the same value almost surely, then $\validity$ is trivial. 
\end{restatable}

\section{Lower Bound on $n$} \label{sec:lower-bound-pki}

In this section, we prove that, for any non-trivial validity property, the condition $n > 2 \cdot t_s + t_a$ is necessary for it to be solvable in the network-agnostic model even if cryptographic setup is available. Our proof will be organized as follows:
we start with a preliminary lemma in \cref{subsec:setup:preliminary}, which focuses on a simplified setting where $n := 2$. In this setting, at most one party can crash if the network is synchronous, and both parties are honest if the network is asynchronous. Our preliminary lemma will show a somewhat counter-intuitive result: roughly, the value decided upon in canonical executions is independent of the honest parties' inputs.
In \cref{subsec:setup:warmup}, we move from the simplified setting with two parties to a warm-up variant of our main proof. We will be working with $n$ parties, but we will only consider the particular case $t_a := 0$. By reducing this case to our preliminary lemma, we show that the condition $n > 2 \cdot t_s$ is necessary in this setting.
Finally, \cref{subsec:setup:dificile} focuses on the general case, showing that $n > 2 \cdot t_s + t_a$ is necessary by reducing to our preliminary lemma. The main proof will be a more general version of the warm-up proof.

We note that the proofs by Civit et al. \cite{PODC24:CGGKPV} rely on reducing any non-trivial validity property to weak validity, enabling lower bounds to focus solely on weak validity.
However, their reduction is invalid for randomized protocols and hence cannot be used in our setting.

\subsection{Preliminary Lemma} \label{subsec:setup:preliminary}

As previously mentioned, our preliminary lemma focuses on a simplified setting with $n := 2$ parties in the network-agnostic model. When the network is synchronous, we allow the adversary to corrupt up to one party ($t_s := 1$). We restrict the adversary's capabilities by only allowing the corrupted party to crash. When the network is asynchronous, both parties are honest ($t_a := 0$). 
Then, we assume a protocol achieving (probabilistic) termination and agreement in this setting. 
We show that, almost surely, once randomness is fixed, all canonical executions decide the same value. Note that there may be a set of  randomnesses such that some canonical executions do not even decide, but the probability of picking a randomness in this set is zero.

\begin{lemma}\label{lemma:agnostic-lemma}
    Assume $n := 2, t_s := 1$, $t_a :=0$, and that corrupted parties are only allowed to crash. Consider a protocol $A$ that achieves probabilistic termination and agreement in this setting. 
    Then, almost surely, there exists a value $v$ such that all canonical executions of $A$ decide $v$ when running.
\end{lemma}
\begin{proof}
Denote the two parties by $\party_1$ and $\party_2$. Let $\omega \in \Omega$. For every input configuration, we will define a finite number of executions using randomness $\omega$, 
including the required canonical executions.
Under the assumption that all executions defined for $\omega$ decide, we show that all canonical executions using randomness $\omega$ decide the same value $v$. From here, our proof proceeds as follows: individually, each defined execution decides almost surely (by construction). The number of executions defined for each input configuration is finite, and the set of input configurations is countable (since $V_\inputt$ is countable). A countable intersection of events holding almost surely holds almost surely. Then, almost surely, all defined executions for $\omega$ decide, i.e., our additional assumption holds almost surely, completing the proof.

Hence, from now on, we fix $\omega \in \Omega$ and only consider executions running with this randomness $\omega$. We assume that all executions we consider decide, and we want to show that all canonical executions decide the same value $v$.


We first consider canonical executions where one of the parties crashes. Let $v_1, v_2 \in V_{\inputt}$ be two arbitrary values.
First, we consider the canonical execution $\varepsilon_1(\omega)$  of $A$ (in a synchronous network) where $\party_1$  is honest and has input $v_1$. $\party_2$ is corrupted and crashes at the beginning of the execution. We have assumed that $\varepsilon_1(\omega)$ is deciding, hence $\party_1$ outputs some value $v_1'$ within a finite amount of time $T_1(\omega)$ (note:~without receiving any messages from $P_2$). 
Second, consider the canonical execution  $\varepsilon_2(\omega)$ of $A$ (again, in a synchronous network) where  $\party_2$ is honest and has input $v_2$. $\party_1$ is corrupted and crashes at the beginning of the execution. We have assumed that $\varepsilon_2(\omega)$ is deciding, hence $\party_2$ outputs some value $v_2'$ within a finite amount of time $T_2(\omega)$. We prove that $v_1' = v_2'$ by defining a third execution $\varepsilon_{1, 2}(\omega)$, but, this time, in the asynchronous model. Here, both $\party_1$ and $\party_2$ are honest with inputs $v_1$ and $v_2$ respectively. The scheduler delays any message between them until time $\max(T_1(\omega), T_2(\omega))$. This way, $\party_1$ cannot distinguish execution $\varepsilon_{1, 2}(\omega)$ from execution $\varepsilon_1(\omega)$, and therefore it outputs $v_1'$ in $\varepsilon_{1, 2}(\omega)$ as well. Similarly, $\party_2$ outputs $v_2'$ in $\varepsilon_{1, 2}(\omega)$. Recall that $A$ achieves agreement, hence $v_1' = v_2'$. 
Note that the argument holds for arbitrary $v_1, v_2 \in V_{\inputt}$, so we get that \emph{all} canonical executions where one of the parties crashes decide the same value, which we denote by $v$.

We now consider canonical executions where none of the parties crashes and show that all such executions also decide $v$. Consider any $v_1, v_2 \in V_{\inputt}$ and the canonical execution $\varepsilon_M(\omega)$ of $A$ where both parties are honest and have inputs $v_1$ and $v_2$. We also define $\varepsilon_{1, 2}(\omega)$ as before. From the previous part, it follows that $\varepsilon_{1, 2}(\omega)$ decides $v$. We will now show that $\varepsilon_M(\omega)$ decides $v$ as well. Based on our assumption, we already know that $\varepsilon_M(\omega)$ decides, so it does so after a finite number of messages $\ell$. For the next part of the proof, we first give an intuitive outline. Roughly, we will build a chain of $\ell + 1$ scenarios between execution $\varepsilon_{1, 2}(\omega)$ and execution $\varepsilon_M(\omega)$: in each scenario $1 \leq m < \ell$, the network is asynchronous, and the scheduler ensures that the first $m$ messages are received synchronously, while all subsequent messages are delayed sufficiently long. Scenarios $0$ and $\ell$ correspond to executions $\varepsilon_{1,2}(\omega)$ and $\varepsilon_M(\omega)$ respectively, and any two consecutive scenarios will be indistinguishable to the party that sent the last message: this will imply that $\varepsilon_M(\omega)$ also decides $v$. In the following, we write this idea formally.

\begin{figure}[H]
\centering
\includegraphics[scale=0.4]{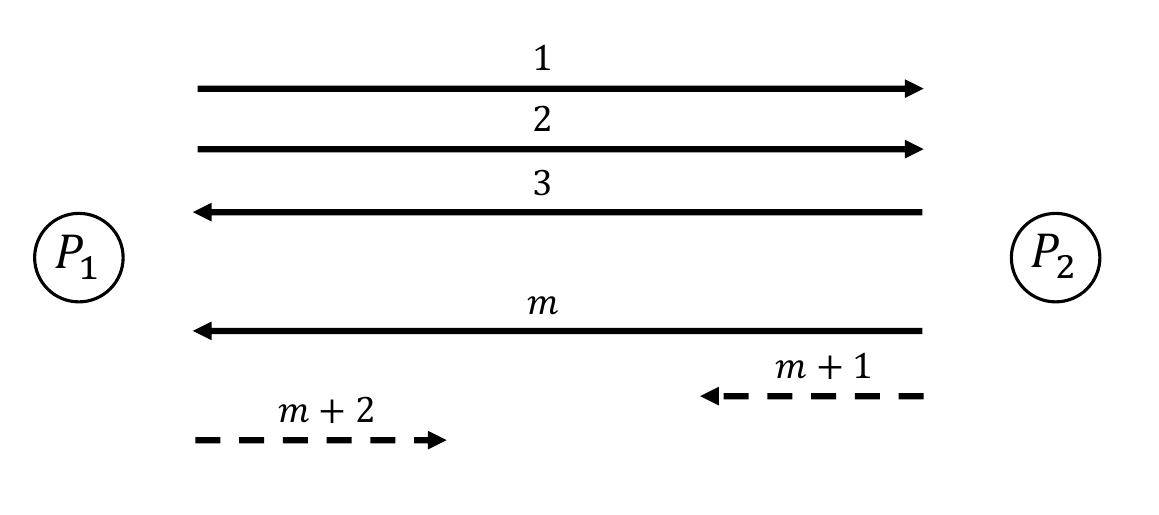}
\caption{Example of execution $\varepsilon_{M,m}(\omega)$: Any messages sent after the first $m$ messages get delayed until after a value has been decided.}
\end{figure}
For $0 \leq m \leq \ell$, we consider an execution $\varepsilon_{M,m}(\omega)$ of $A$ in the asynchronous model: both $P_1$ and $P_2$ are honest, and they have inputs $v_1$ and $v_2$ respectively. 
In execution $\varepsilon_{M,0}(\omega)$, the scheduler uses the same strategy as that of execution $\varepsilon_{1, 2}(\omega)$:
the scheduler delays any message in execution $\varepsilon_{M,0}(\omega)$ until time $T = \max(T_1(\omega), T_2(\omega))$. Hence, 
$\party_1$ and $\party_2$ cannot distinguish between executions $\varepsilon_{M,0}(\omega)$ and $\varepsilon_{1,2}(\omega)$, implying that
$\varepsilon_{M,0}(\omega)$ decides $v$ by some time $T_0$. In execution $\varepsilon_{M,1}(\omega)$, the scheduler allows the first message to be delivered after exactly $\Delta$ time. Assume without loss of generality that this message is sent by $\party_1$. All further messages are delayed until time $T_1 > T_0$ (we define the exact time later): hence, $\party_1$ cannot distinguish from this execution and execution $\varepsilon_{M,0}(\omega)$, therefore it outputs $v$ by time $T_0$, as in execution $\varepsilon_{M,0}(\omega)$. $\party_2$, on the other hand, can distinguish between $\varepsilon_1(\omega)$ and $\varepsilon_{M,0}(\omega)$: however, it cannot distinguish between $\varepsilon_{M,1}(\omega)$ and an execution where the network is, in fact, synchronous, and $\party_1$ has crashed. If $\party_1$  has crashed, $\party_2$ has to output eventually, hence by time $T_1'$. Hence, the scheduler makes sure to delay all messages until time $T_1 > \max(T_0, T_1')$. Consequently, $P_2$ has to output $v$ in $\varepsilon_{M,1}(\omega)$ by time $T_1$.
The next executions are defined similarly: in execution $\varepsilon_{M, m}(\omega)$ with $m > 0$, the scheduler allows the first $m$ messages to be delivered after exactly $\Delta$ units of time. The remaining messages are delivered sufficiently late to ensure that the last message's sender is unable to distinguish between $\varepsilon_{M, m}(\omega)$ and execution $\varepsilon_{M,m-1}(\omega)$. Thus, the last message's sender outputs $v$ in execution $\varepsilon_{M, m}(\omega)$, which forces the other honest party to also output $v$.


We remark that executions $\varepsilon_{M,\ell}(\omega)$ and $\varepsilon_{M}(\omega)$ are indistinguishable, so $\varepsilon_{M}(\omega)$ will decide $v$. Since this holds for arbitrary values $v_1, v_2$, we have obtained that all canonical executions where no party crashes also output $v$, hence completing the proof.
\end{proof}

\subsection{Warm-up: $t_a := 0$} \label{subsec:setup:warmup}

As a warm-up towards the main result, we focus on the particular case $t_a := 0$, and show that the condition $n > 2 \cdot t_s$ is necessary. Intuitively, when $n \leq 2 \cdot t_s$, one cannot distinguish between a scenario where half of the parties crash in a synchronous network, and a scenario where half of the parties are honest but delayed in an asynchronous one. This way, the presence of the asynchronous case, even with no corruptions, allows two disjoint sets of honest parties to only run the protocol within their own set. Since the two sets run the protocol independently, the honest parties agree on a valid value only if the given validity property is trivial.
Note that this is the case even for weak validity, which can be solved in the synchronous model for up to $t_s < n$ corruptions. Our result implies that expecting a synchronous protocol to provide guarantees when it runs in a corruption-free asynchronous network impacts the overall resilience.


\begin{theorem} \label{thm:warm-up}
    Assume $t_s > 0$ and consider a validity property $\validity$. If $n = 2 \cdot t_s$ and there is a $(t_s, 0)$-secure $\ba$ protocol solving $\validity$, then $\validity$ is trivial.
\end{theorem}
\begin{proof}

Consider a (possibly randomized) $(t_s, 0)$-secure $\ba$ protocol $\Pi$ that solves $\val$ when $n = 2 \cdot t_s$. 
Let $I_1$ and $I_2$ be two arbitrary maximal input configurations.
We show that the canonical executions of $I_1$ and $I_2$ almost surely decide on the same output when run using the same randomness. Then, \cref{lemma:technical-end} ensures that $\validity$ is trivial.

\begin{figure}[H]
\centering
\includegraphics[scale=0.13]{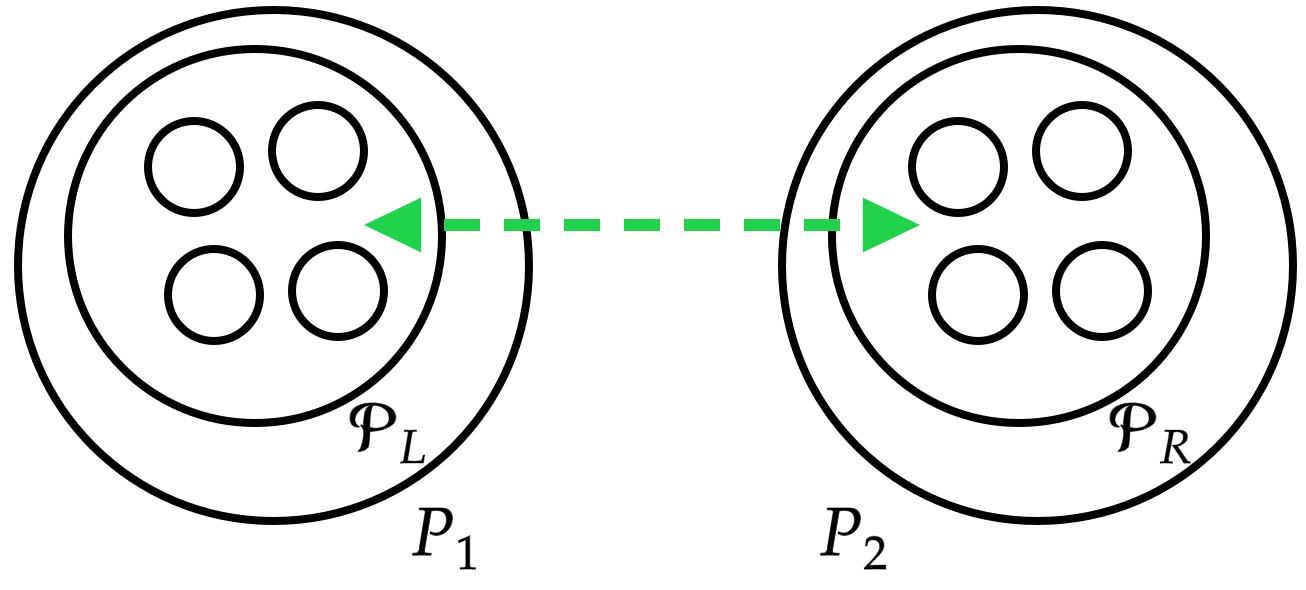}
\caption{$P_1$ simulates all parties of $\PS_L$ while $P_2$ simulates all parties of $\PS_R$. All the parties of $\PS$ communicate between one another not knowing they are being simulated.}\label{fig:simulation1}
\end{figure}

We use $\Pi$ to build a $2$-party protocol $A$ that matches the setting described in \cref{lemma:agnostic-lemma}. 
For protocol $A$, we consider the input set $\{0,1\}$ and output set $V_{\outputt}$, and we denote the two parties running $A$ by $P_1$ and $P_2$. Since $n = 2 \cdot t_s$, we may partition the set of $n$ parties $\PS$ into two sets $\PS_L$ and $\PS_R$ of $t_s$ parties each.

Then, as shown in Figure~\ref{fig:simulation1}, $\party_1$ will simulate the parties in $\PS_L$, while $\party_2$ will simulate the parties in $\PS_R$.
Concretely, in protocol $A$, $P_1$ proceeds as follows:

\begin{itemize}
    \item $P_1$ simulates all $t_s$ parties of $\PS_L$ running $\Pi$.
    \item If $P_1$ has input $0$, then the simulated parties in $\PS_L$ use their inputs from $I_1$. Otherwise, they use their inputs from $I_2$.
    \item Messages between the simulated parties in $\PS_L$ are received after exactly $\Delta$ units of time, and $\party_1$ forwards to $\party_2$ every message sent by a simulated party in $\PS_L$ to a party in $\PS_R$. Once $\party_2$ receives this message, it immediately forwards it to the simulated receiver in $\PS_R$.
    \item As soon as a party in $\PS_L$ decides a value, $P_1$ decides this value. $P_1$ continues forwarding messages as described above.
\end{itemize}

$P_2$ proceeds identically to $P_1$, switching $\PS_L$ and $\PS_R$.

Note that running $A$ in an asynchronous network where both $\party_1$ and $\party_2$ are honest corresponds to running $\Pi$ in an asynchronous network where all $n$ parties are honest. In addition, running $A$ in a synchronous network where at most one party may crash corresponds to running $\Pi$ in a synchronous network where at most $t_s$ parties may crash. Then, since $\Pi$ is a $(t_s, 0)$-secure $\ba$ protocol, $A$ achieves probabilistic termination and agreement in the setting described in \cref{lemma:agnostic-lemma}. 
Then, applying \cref{lemma:agnostic-lemma}, we obtain that: almost surely, there is a value $v$ such that all canonical executions of $A$ decide on the same value $v$. Moreover, the canonical execution of $A$ with input values $(0,0)$ matches the canonical execution of $I_1$ for $\Pi$. Similarly, the canonical execution of $A$ with input $(1,1)$ matches the canonical execution of $I_2$. As a consequence, canonical executions of $I_1$ and $I_2$ decide the same value almost surely. Using \cref{lemma:technical-end}, we can therefore conclude that $\validity$ is trivial. 
\end{proof}
\subsection{General Result} \label{subsec:setup:dificile}

We are now ready to prove the final statement of this section, presented below.
\begin{theorem}\label{theo:agnostic-lower-bound}
    Consider a validity property $\validity$, and $t_s, t_a$ such that $t_s > 0$ and $t_s \geq t_a$.
    If there is a $(t_s, t_a)$-secure $\ba$ protocol solving $\validity$ when $n \leq 2 \cdot t_s + t_a$, then $\validity$ is trivial.
\end{theorem}
\begin{proof}
    Assume $n = 2 \cdot t_s + t_a$ and that there is a $(t_s, t_a)$-secure $\ba$ protocol $\Pi$ solving $\validity$. We  consider two input configurations $I_1$ and $I_2$ such that $\parties(I_1) = \parties(I_2) = \PS$ (i.e., all parties are honest). We show that, almost surely, the canonical executions of $I_1$ and $I_2$ decide on the same output. Then, \cref{lemma:technical-end} ensures that $\validity$ is trivial. Similarly to the proof of \cref{thm:warm-up}, we use $\Pi$ to build a two-party protocol $A$ that matches the setting of \cref{lemma:agnostic-lemma}. For $A$, we consider the input space $\{0, 1\}$ and the output space $V_{\outputt}$.

    This time, we partition $\PS$ into three sets $\PS_L$, $\PS_M$ and $\PS_R$ such that $|\PS_L| = |\PS_R| = t_s$ and $|\PS_M| = t_a$.
    A crucial difference from the proof of \cref{thm:warm-up} is that, as shown in \Cref{fig:dificile-figure}, each party in $\PS$ simulates its own copy
    of the parties in $\PS_M$.
    Our construction will ensure that, when $A$ runs in the synchronous setting, the two simulated copies of each party in $\PS_M$ will be in the exact same state at any point, maintaining the guarantees of $\Pi$ when at most $t_s$ of the parties are corrupted. Meanwhile, in the asynchronous setting, the copies of the $t_a$ parties in $\PS_M$ may be in different states, but $\Pi$ is able to tolerate $t_a$ byzantine parties in this case. 


    \begin{figure}[ht]
    \centering
    \includegraphics[scale=0.13]{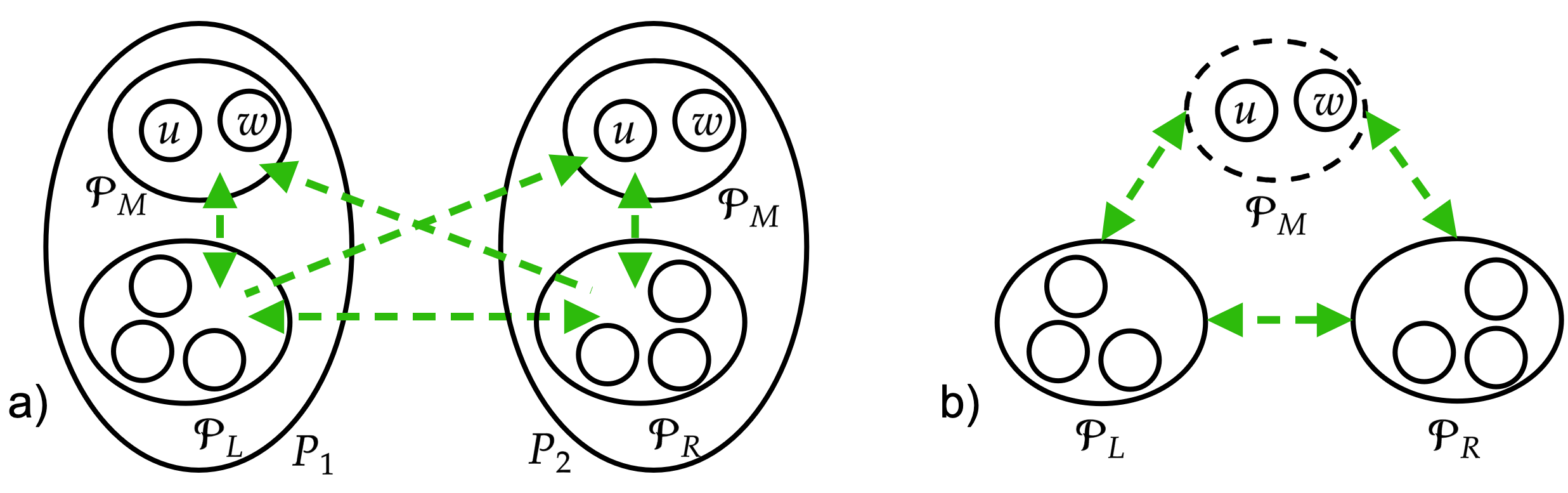}
    \caption{a) $P_1$ simulates all parties of $\PS_L$ and its own copies of the parties in $\PS_M$. $P_2$ simulates all parties of $\PS_R$, and its own copy of the parties $\PS_M$. The arrows show which simulated group of parties can send messages to which other group. \\ b) This is how the network looks like from the point of view of a party in $\PS_R$ or $\PS_M$: they are unaware of the second set $\PS_M$ being simulated in parallel.}\label{fig:dificile-figure}
    \end{figure}

    Concretely, in protocol $A$, party $P_1$ proceeds as follows:
    \begin{itemize}[nosep]
        \item $P_1$ simulates all $t_s + t_a$ parties of $\PS_L \cup \PS_M$ running $\Pi$.
        \item If $P_1$ has input $0$, then parties in $\PS_L \cup \PS_M$ take their input from $I_1$. Otherwise, they take their input from $I_2$.
        \item As depicted in Figure \ref{figure:magic-messages}, the messages sent between the parties simulated by $P_1$ are exchanged as if all $t_s + t_a$ parties are running independently: each such message is delivered after exactly $\Delta$ units of time. Additionally, once a simulated party in $\PS_L$ sends a message to a simulated party in $\PS_M$, $P_1$ immediately forwards this message to $P_2$ as well. Once $P_2$ receives this message, it immediately forwards it to its own simulated receiver in $\PS_M$. The message delay here depends on the type of network that $A$ is running in.
        \item Messages from a party in $\PS_L$ to $\PS_R$ are sent from $P_1$ to $P_2$. $P_2$ then forwards each such message to its simulated receiver in $\PS_R$. 
        \item $P_1$ discards any messages sent by the simulated parties in  $\PS_M$ to parties in $\PS_R$.
        \item As soon as a party in $\PS_L$ (but not $\PS_M$) decides a value, $P_1$ decides this value. $P_1$ continues forwarding messages as described above until all its simulated parties terminate.
    \end{itemize}
    The behavior of $P_2$ is the same as the behavior of $P_1$, switching $\PS_L$ and $\PS_R$.


    \begin{figure}[h]
    \centering
    \includegraphics[scale=0.13]{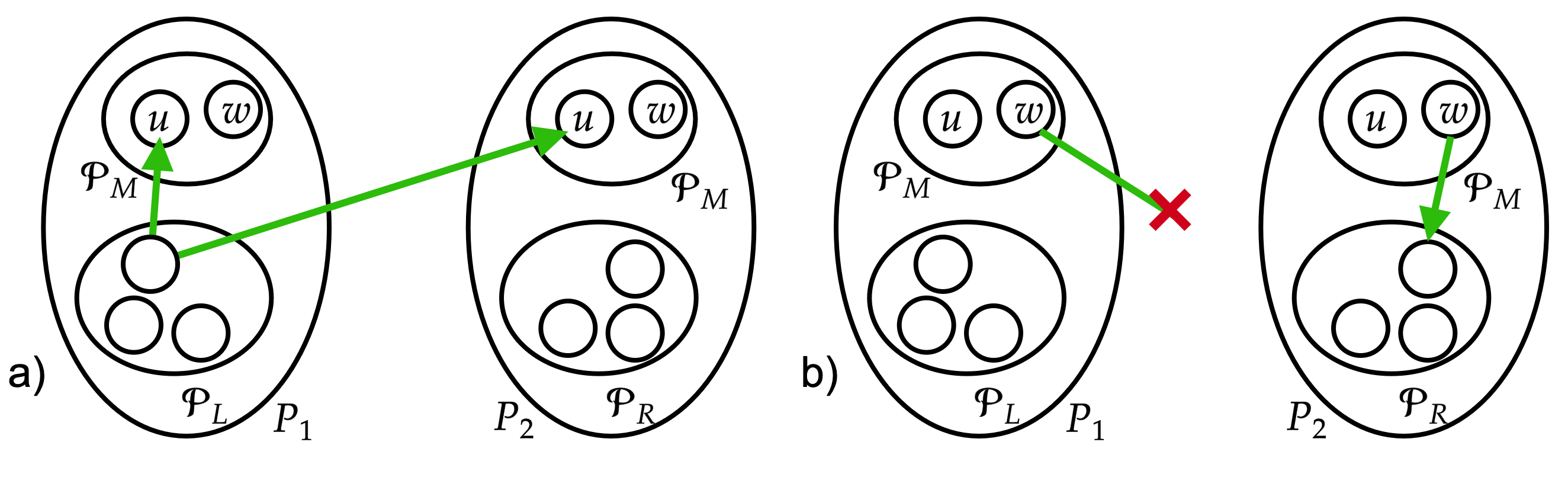}
    \caption{Examples of how different messages are handled. \\
    a) If a simulated party in $\PS_L$ wants to send a message to a party in $\PS_M$, then two identical messages are actually sent: the first one to the copy of the receiver simulated by $P_1$, and the second to the copy simulated by in $P_2$.\\
    b) If a party in $\PS_M$ wants to send a message to a party in $\PS_R$, if the two parties are simulated by the same entity (here $P_2$), then the message gets received as expected. Otherwise, the message is discarded and the party in $\PS_R$ never receives it.} \label{figure:magic-messages}
    \end{figure}

    We now need to analyze $A$ in the setting described by \cref{lemma:agnostic-lemma}. 
    Running $A$ in an asynchronous network where both parties are honest corresponds to running $\Pi$ in a network where the communication between parties in $\PS_L$ and $\PS_R$ is asynchronous. While the simulated copies in $\PS_M$ are not necessarily consistent, $\Pi$ is resilient against $t_a = \abs{\PS_M}$ byzantine corruptions, and therefore maintains its guarantees. Hence, $A$ achieves agreement and probabilistic termination in this setting.
    It remains to show that $A$ also achieves these guarantees in a synchronous network where one of the two parties may crash. Settings where both $\party_1$ and $\party_2$ are honest correspond to running $\Pi$ in a synchronous network where all $n$ parties are honest. We note that, in this case, for each party in $\PS_M$, the copy simulated by $\party_1$ and the copy simulated by $\party_2$ maintain the same state at all times.
    Settings where at most one of $\party_1$ and $\party_2$ may crash correspond to running $\Pi$ in a synchronous setting where the $t_s$ parties meant to be simulated only by the party crashing in $A$ are corrupted. As $\Pi$ tolerates $t_s$ corruptions, it maintains its guarantees. Hence, $A$ achieves agreement and probabilistic termination in this setting as well.
    

    We conclude the proof in an identical manner to the proof of \cref{thm:warm-up}. Applying Lemma~\ref{lemma:agnostic-lemma}, we obtain that, almost surely, all canonical executions of $A$ decide on the same value.
    Moreover, the canonical execution of $A$ with input values $(0,0)$ matches the canonical execution of $I_1$ for $\Pi$. Similarly, the canonical execution of $A$ with input $(1,1)$ matches the canonical execution of $I_2$.
    As a consequence, canonical executions of $I_1$ and $I_2$ decide the same value almost surely. Using \cref{lemma:technical-end}, we can therefore conclude that $\validity$ is trivial.
\end{proof}
\section{Lower Bound on $n$ Without Cryptographic Setup}

The previous section has proven that the condition $n > 2 \cdot t_s + t_a$ is necessary regardless of whether a cryptographic setup is available or not. We now focus on settings without cryptographic setup and prove an even stronger condition. We show that, in such settings, the condition $n > 3 \cdot t_s$ is necessary even in the synchronous model. Since a protocol achieving $(t_s, t_a)$-secure $\ba$ in the network-agnostic model also achieves $t_s$-secure $\ba$ in the synchronous model, this bound immediately applies to the network-agnostic model. We add that our result extends the requirement of $n > 3 \cdot t_s$ provided by Civit et al. \cite{PODC24:CGGKPV} for synchronous deterministic protocols to cover randomized protocols as well.

\subsection{Preliminary Lemma}

Similarly to the outline of \cref{sec:lower-bound-pki}, we first consider a setting with three parties, out of which at most one is byzantine. Afterwards, we focus on the general case. \Cref{lemma:3-parties} is an improvement over the result of Fischer, Lynch, and Merritt for weak validity \cite{PODC:FisLynMer85}: the result of \cite{PODC:FisLynMer85} only applies for weak validity and protocols must always decide in a finite amount of time, which is a lot stronger than probabilistic termination. Our proof uses the main core idea but improves it to lift these restrictions. Roughly,  we will be running $A$ in a larger ring containing multiple copies of each party, as depicted in Figure \ref{figure:shinier-ring}.

\begin{restatable}[h]{figure}{RingFigure}
\centering
\includegraphics[scale=0.13]{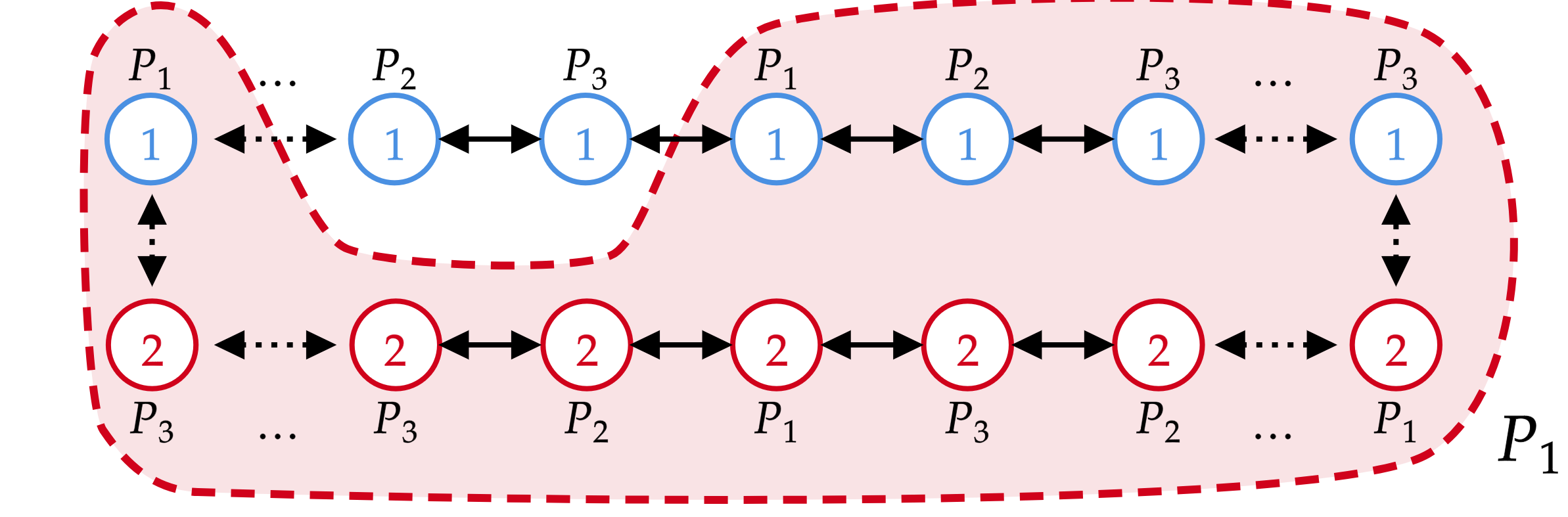}
\caption{Defining the behavior of a byzantine party (here $P_1$).}\label{figure:shinier-ring}
\end{restatable}

 The ring is constructed from two canonical executions with different input configurations. Parties adjacent in this ring cannot distinguish between the ring and the original setting of three parties, as the third party may be byzantine and simulate the rest of the ring.
This forces parties adjacent in the ring to output the same value, which implies that the two original executions lead to the same output. We defer the formal proof to Appendix~\ref{appendix:ring}.

\begin{restatable}{lemma}{RingLemma} \label{lemma:3-parties}
    Consider $n := 3$ parties in a synchronous network, and assume a protocol $A$ that achieves (probabilistic) termination and agreement in this setting even when up to one party is byzantine.
    Then, almost surely, all canonical executions of $A$ decide the same value.
\end{restatable}

\subsection{General Result}

To prove our main result in this setting, we use a strategy similar to the proof of \cref{thm:warm-up}. Concretely, we assume an $n$-party protocol that achieves $\val$ whenever $n \geq 3 \cdot t_s$ and use it to construct a three-party protocol that contradicts Lemma \ref{lemma:3-parties}. We defer the formal proof to Appendix~\ref{appendix:ring-extended}.
\begin{restatable}{theorem}{RingExtended}\label{theo:t-sync}
    Assume $t_s > 0$ and consider a validity property $\validity$. If there is a $t_s$-secure $\ba$ protocol solving $\validity$ in the synchronous model when no cryptographic setup is available and $n \leq 3 \cdot t_s$, then $\validity$ is trivial.
\end{restatable}

Then, as any lower bound in the synchronous model also holds in the network-agnostic model, \cref{theo:t-sync} immediately implies the following corollary.
\begin{corollary}\label{corollary:t-sync}
    
    Consider a validity property $\validity$, and let $t_s, t_a$ such that $t_s > 0$ and $t_s \geq t_a$.
    If there is a $(t_s, t_a)$-secure $\ba$ protocol solving $\validity$ when no cryptographic setup is available and $n \leq 3 \cdot t_s$, then $\validity$ is trivial.
    
\end{corollary}
\section{Similarity Condition}

The lower bounds on $n$ presented in the previous sections are indeed necessary, but not yet sufficient. 
We may instantiate, for instance, the input space as the (finite) set of vertices of a (publicly known) labeled graph with maximum clique size $\omega \geq 3$. We consider convex validity, under the so-called \emph{monophonic convexity}. For this example, network-agnostic $\ba$ requires the stronger lower bound $n > \max(\omega \cdot t_s, \omega \cdot t_a + t_s, 2 \cdot t_s + t_a)$ \cite{constantinescu_ghinea_convex_consensus}. Roughly, if any of the conditions $n > \omega \cdot t_s$ and $n > \omega \cdot t_a + t_s$ fails to hold, one can find scenarios where the simple presence of byzantine parties (following the protocol correctly, with inputs of their choice) prevents the honest parties from obtaining a valid output. In this section, we prove the need of one more condition that captures these validity-dependent requirements, and enables the honest parties to find a valid output even if their view over the honest inputs is not accurate. This additional condition matches the \emph{similarity} condition of \cite{PODC23:CGGKV} for the partially synchronous model, and the \emph{containment} condition of \cite{PODC24:CGGKPV} in the synchronous model.

We need to establish a few notions. The first is the notion of \emph{neighbors} of an input configuration. The neighbors of $I$, denoted by $\neighbors(I)$, are the input configurations $J$ such that the parties in $\parties(I) \cap \parties(J)$ hold the same input values in $I$ and $J$.
Formally, $\neighbors(I) := \{ J \in \inputconfigs: \  \forall \party \in \PS, \ \text{if} \ (v_1, \party) \in I \ \text{and} \ (v_2, \party) \in J \ \text{then} \ v_1 = v_2\}$.
The definition of neighbors is symmetric (if $J \in \neighbors(I)$ then $I \in \neighbors(J)$).

The second notion is that of \emph{similar configurations} of an input configuration $I$, denoted by $\similar(I)$. These are input configurations that may be indistinguishable from $I$. In the synchronous model, these are configurations $J \in \neighbors(I)$ such that $J \subseteq I$. Roughly, this models scenarios where some of the parties are corrupted, but follow the protocol correctly with inputs of their own choice. In the asynchronous model, these are configurations $J \in \neighbors(I)$ containing $n - t_a$ parties: this  additionally takes into account that at most $t_a$ honest parties may be isolated due to network delays. Hence, we define $\similar(I)$ as
$
    \similar(I) = \{ J \in \neighbors(I): J \subseteq I \} \bigcup \{ J \in \neighbors(I): |J| \geq n - t_a \}.
$



We may now define the similarity condition, which \cref{lemma:similarity} proves to be necessary for the network-agnostic model. 
\cref{lemma:similarity} adapts Theorem 2 of \cite{PODC23:CGGKV} and Lemma 8 of \cite{PODC24:CGGKPV} to the network-agnostic model, and it relies on standard indistinguishability arguments. We defer the proof to Appendix \ref{appendix:similarity}.

\begin{definition}[Similarity condition]
We say that a validity property $\validity$ satisfies the \emph{similarity condition} if there is a Turing-computable function $\sigma : \inputconfigs \mapsto V_O$ such that, for any input configuration $I \in \inputconfigs$, $\sigma(I) \in \bigcap_{J \in \similar(I)} \validity(J)$.
\end{definition}

Note that the existence of such a function $\sigma$ also implies  $\forall I \in \inputconfigs, \bigcap_{J \in \similar(I)} \validity(J) \neq \emptyset$. 

\begin{restatable}{lemma}{SimilarityLemma} \label{lemma:similarity}
    If a validity property $\validity$ is solvable in the network-agnostic model, then $\validity$ satisfies the similarity condition.
\end{restatable}

It is worth noting that, in the proof of this lemma, the deterministic Turing function $\sigma$ is defined based on a randomized protocol that solves $\validity$. This is justified because, as discussed in the proof, when time complexity is not taken into account, probabilistic Turing machines are as expressive as deterministic ones.

\section{Sufficiency and Main Result}

We now show that the conditions presented in the previous sections are not only necessary, but also sufficient, hence proving our main result, stated below.
\begin{theorem}\label{thm:main-thm}
    Assume a non-trivial validity condition $\validity$. Then, there is a $(t_s, t_a)$-secure protocol solving $\validity$ if and only if the following conditions hold:
    \begin{itemize}[nosep]
        \item $\validity$ satisfies the similarity condition.
        \item $n > 3 \cdot t_s$ or, if PKI is available, $n > 2 \cdot t_s + t_a$.
    \end{itemize}
\end{theorem}

The remainder of the section describes a protocol that matches our necessary conditions, as stated in the lemma below. 
\cref{thm:main-thm} then follows from combining Lemma  \ref{lemma:matching-protocol} with the requirements discussed previously: the lower bounds on $n$, described in \cref{theo:agnostic-lower-bound} and \cref{corollary:t-sync}, plus the similarity condition, proven to be necessary in \cref{lemma:similarity}.
\begin{lemma}\label{lemma:matching-protocol}
    Assume a validity condition $\validity$ that satisfies the similarity condition. Then, if $n > 3 \cdot t_s$, or, assuming that PKI is available, if $n > 2 \cdot t_s + t_a$, there is a $(t_s, t_a)$-secure  $\ba$ protocol solving $\validity$.
\end{lemma}

Our construction behind Lemma \ref{lemma:matching-protocol} generalizes the network-agnostic $\ba$ protocol of \cite{constantinescu_ghinea_convex_consensus} that solves convex validity. The parties distribute their input values using a protocol achieving \emph{Agreement on a Core-Set} ($\acs$), which provides them with an identical view over the inputs, i.e., with a potential input configuration. This is a (randomized) communication primitive enabling identical views in the pure asynchronous model \cite{STOC:BCG93,PODC:BKR94}. We make use of the $\acs$ definition of \cite{constantinescu_ghinea_convex_consensus}, included below, which differs from the standard definition by providing stronger properties in the synchronous model: it additionally ensures that all honest inputs are included in the common view if the network is synchronous.
This property is essential for matching the higher resilience threshold $t_s$ in a synchronous network. This way, the parties can apply a deterministic decision over the view agreed upon in $\acs$ and obtain an output.

\begin{definition} \label{definition:acs}
 Let $\Pi$ be a protocol where every party $\party_i$ holds an input $v_i$ and outputs a set $I_i \in \inputconfigs$ consisting of at least $n - t_s$ value-sender pairs. We consider the following properties in addition to those presented in \cref{sec:preliminaries}:
    \begin{itemize}[nosep]
    \item \textbf{Integrity}: If $\party_i$ and $\party_j$ are both honest and $\party_i \in \parties(I_j)$, then $(\party_i, v_i) \in I_j$.
            
    \item \textbf{Honest Core}: If an honest party outputs $I$, then $\parties(I)$ contains all honest parties.
    \end{itemize}
    Then, we say that $\Pi$ is a $(t_s, t_a)$-secure $\acs$ protocol if it achieves the following:
    \begin{itemize}[noitemsep,nolistsep]
        \item Probabilistic termination, agreement, integrity, and honest core when running in a synchronous network where up to $t_s$ parties are corrupted;
        \item Probabilistic termination, agreement, and integrity when running in an asynchronous network where up to $t_a$ parties are corrupted.
    \end{itemize}
\end{definition}

We make use of the $\acs$ construction of \cite{constantinescu_ghinea_convex_consensus}, described by the result below.
\begin{theorem}[\cite{constantinescu_ghinea_convex_consensus}]\label{thm:convex-world:acs}
    Consider $n, t_s, t_a$ such that $t_a \leq t_s$. If  $3 \cdot t_s < n$, or, if PKI is available and $2 \cdot t_s + t_a < n$, there is a $(t_s, t_a)$-secure $\acs$ protocol $\Pi_{\acs}$.
\end{theorem}


We may now present the proof of \cref{lemma:matching-protocol}, describing our protocol.

\begin{proof}[Proof of \cref{lemma:matching-protocol}]
    Our $\ba$ protocol solving $\validity$ proceeds as follows: the parties distribute their input values using the protocol $\Pi_\acs$ described in \cref{thm:convex-world:acs}. $\Pi_\acs$ provides the parties with the same output set $I \in \inputconfigs$ of at least $n - t_s$ value-sender pairs, representing a potential input configuration. Once the parties obtain this set, they output $\sigma(I)$, where $\sigma$ is the Turing-computable function provided by $\validity$ satisfying the similarity condition.

    \begin{dianabox}{$(t_s, t_a)$-secure $\ba$ solving $\validity$}
    	\algoHead{Code for party $\party_i$ with input $v_i$}
    	\begin{algorithmic}[1]
                \State Join $\Pi_\acs$ with input $v_i$ and obtain the output set $I$.
                \State Output $\sigma(I)$ and terminate.
    	\end{algorithmic}
    \end{dianabox}
 
     Regardless of whether the network is synchronous or asynchronous, since $\Pi_{\acs}$ achieves probabilistic termination, our $\ba$ protocol achieves probabilistic termination as well. In addition, since $\Pi_{\acs}$ achieves agreement, the parties compute their output identically, and therefore our $\ba$ protocol achieves agreement.

     It remains to prove that the honest parties' output is in $\validity(H)$, where $H$ denotes the (actual) input configuration (containing only the honest parties and their inputs). Since $\sigma(I) \in \bigcap_{J \in \similar(I)} \validity(J)$, it will be sufficient to show that $H \in \similar(I)$.
     Regardless of the type of network, the integrity property of $\Pi_{\acs}$ ensures that $H \in \neighbors(I)$. If the network is asynchronous, $\abs{H} \geq n - t_a$ and therefore $H \in \similar(I)$.
     Otherwise, if the network is synchronous,
     the honest core property of $\Pi_{\acs}$ ensures that honest parties and their inputs are included in $I$. Then, $H \subseteq I$ and therefore $H \in \similar(I)$ in this case as well. Thus, as $\sigma(I) \in \bigcap_{H \in \similar(I)} \validity(H)$, we get that the value agreed upon is in $\validity(H)$, which proves the validity of the protocol and concludes the proof.
\end{proof}




\section{On Uncountable Input and Output Sets} \label{section:uncountable-stuff}

In this section, we make a remark on the importance of assuming that either the input set $V_{\inputt}$ is countable, the output set $V_{\outputt}$ is countable, or that the probability space $\Omega$ has some limitations. For our proofs, we assumed that $V_{\inputt}$ is countable, and it is possible to adapt all arguments for the other case. We also assumed that only a finite amount of random bits can be retrieved at a given instant. In contrast, perhaps surprisingly, not only our arguments cannot be adapted for the case where both $V_{\inputt}$ and $V_{\outputt}$ are uncountably infinite, but, in the following, we describe a brief counterexample to the fact that $n > 2 \cdot t_s + t_a$ is a requirement for a non-trivial validity property to be solvable. This example requires a different probability space, allowing a real number to be sampled in a single step.

For the uncountably infinite sets $V_{\inputt} = V_{\outputt} := [0, 1]$, we define a validity property $\validity^{\star}$ as follows: if all honest parties hold the same input $v$, they may output any value in $[0, 1]$ except for $v$; otherwise, they may output any value in $[0, 1]$.
Hence, $\validity^{\star}(I) = [0, 1] \setminus \{v : \textit{all parties in } I \textit{ have input v}\} $. 

Note that this validity property is non-trivial as $\bigcap_{v \in [0, 1]} \left( [0, 1] \setminus \{v\} \right) = \emptyset$.

This property escapes our requirement of $n > 2 \cdot t_s + t_a$ if a very strong shared randomness setup is provided. Under the assumption that the parties are provided with a shared uniform random value $s \in [0, 1]$, we can design a $(n - 1, n - 1)$-secure $\ba$ protocol where parties do not even communicate. We present the protocol below.

\begin{dianabox}{$(n - 1, n - 1)$-secure $\ba$ solving $\validity^{\star}$}
	\algoHead{Code for party $\party_i$ with input $v_i$}
	\begin{algorithmic}[1]
            \State Let $s :=$ the shared uniform random value in $[0, 1]$.
		\State If $v_i = s$, never terminate. Otherwise, output $s$ and terminate.
	\end{algorithmic}
\end{dianabox}

The non-terminating executions of this protocol are those where $s$ matches some honest party's input, hence they occur with probability $0$. That is, all honest parties decide on an output almost surely, which satisfies our probabilistic termination definition requirement. If two honest parties obtain outputs, then they both output the same value $s$. Hence, the agreement property holds. Finally, since no honest party outputs its own input, our protocol solves $\validity^\star$. 
If either $V_{\inputt}$ or $V_{\outputt}$ is countable, our counterexample fails. We note that restrictions on the protocol's setup may also be promising for extending our results to settings where both $V_{\inputt}$ and $V_{\outputt}$ uncountable: if $s$ comes from a finite domain, then the probability of obtaining a non-terminating execution would not be $0$.
\section{Conclusions and Future Work}
We investigated the conditions that a validity property needs to satisfy in order to be solvable in the network-agnostic model and established the necessary and sufficient conditions. Our results demonstrate that solving a non-trivial validity property $\validity$ requires
\begin{enumerate*}[label=(\roman*)]
\item that $\validity$ satisfies the \emph{similarity condition}, and
\item that $n > 2 \cdot t_s + t_a$ assuming a public key infrastructure, or $n > 3 \cdot t_s$ otherwise. 
\end{enumerate*}
Further, we provided a universal protocol that solves a given validity property whenever these established conditions are met.
Our characterization follows the line of works of \cite{PODC23:CGGKV,PODC24:CGGKPV} focusing on the partially synchronous model and on the synchronous model. At the same time, it generalizes prior results on when network-agnostic $\ba$ can be achieved \cite{TCC:BKL19, constantinescu_ghinea_convex_consensus} from fixed validity properties to arbitrary validity properties.

While our work provides a complete answer for solvability, we leave a number of exciting problems open. 
Our results and approach do not hold if we allow the protocol to fail with some probability $\varepsilon > 0$. Focusing on this weaker setting would likely lead to more resilient protocols.
Future works could also extend our characterization to cover settings where both $V_{\inputt}$ and $V_{\outputt}$ are uncountable sets and consider proving message complexity lower bounds. Other promising directions would aim to improve the efficiency of our universal protocol, or to generalize our results further to network-dependent validity properties (that allow weaker guarantees if the network is asynchronous) \cite{OPODIS:CGHWW23}, or to weaker agreement definitions \cite{constantinescu_ghinea_convex_consensus}.



\begin{thebibliography}{10}

\bibitem{OPODIS:AAD04}
Ittai Abraham, Yonatan Amit, and Danny Dolev.
\newblock Optimal resilience asynchronous approximate agreement.
\newblock In {\em Principles of Distributed Systems}, pages 229--239, Berlin, Heidelberg, 2005.
\newblock \href {https://doi.org/10.1007/11516798_17} {\path{doi:10.1007/11516798_17}}.

\bibitem{PODC:ApChCh22}
Ananya Appan, Anirudh Chandramouli, and Ashish Choudhury.
\newblock Perfectly-secure synchronous mpc with asynchronous fallback guarantees.
\newblock In {\em Proceedings of the 2022 ACM Symposium on Principles of Distributed Computing}, PODC'22, page 92–102, New York, NY, USA, 2022. Association for Computing Machinery.
\newblock \href {https://doi.org/10.1145/3519270.3538417} {\path{doi:10.1145/3519270.3538417}}.

\bibitem{STOC:BCG93}
Michael Ben-Or, Ran Canetti, and Oded Goldreich.
\newblock Asynchronous secure computation.
\newblock In {\em Proceedings of the Twenty-Fifth Annual ACM Symposium on Theory of Computing}, STOC '93, page 52–61, New York, NY, USA, 1993. Association for Computing Machinery.
\newblock \href {https://doi.org/10.1145/167088.167109} {\path{doi:10.1145/167088.167109}}.

\bibitem{PODC:BKR94}
Michael Ben-Or, Boaz Kelmer, and Tal Rabin.
\newblock Asynchronous secure computations with optimal resilience (extended abstract).
\newblock In {\em Proceedings of the Thirteenth Annual ACM Symposium on Principles of Distributed Computing}, PODC '94, page 183–192, New York, NY, USA, 1994. Association for Computing Machinery.
\newblock \href {https://doi.org/10.1145/197917.198088} {\path{doi:10.1145/197917.198088}}.

\bibitem{TCC:BKL19}
Erica Blum, Jonathan Katz, and Julian Loss.
\newblock Synchronous consensus with optimal asynchronous fallback guarantees.
\newblock In {\em Theory of Cryptography Conference}, pages 131--150. Springer, 2019.
\newblock \href {https://doi.org/10.1007/978-3-030-36030-6_6} {\path{doi:10.1007/978-3-030-36030-6_6}}.

\bibitem{Crypto:BLL20}
Erica Blum, Chen-Da Liu-Zhang, and Julian Loss.
\newblock Always have a backup plan: Fully secure synchronous mpc with asynchronous fallback.
\newblock In {\em Advances in Cryptology -- CRYPTO 2020}, pages 707--731, 2020.
\newblock \href {https://doi.org/10.1007/978-3-030-56880-1_25} {\path{doi:10.1007/978-3-030-56880-1_25}}.

\bibitem{BOUZID20103154}
Zohir Bouzid, Maria~Gradinariu Potop-Butucaru, and Sébastien Tixeuil.
\newblock Optimal byzantine-resilient convergence in uni-dimensional robot networks.
\newblock {\em Theoretical Computer Science}, 411(34):3154--3168, 2010.
\newblock \href {https://doi.org/10.1016/j.tcs.2010.05.006} {\path{doi:10.1016/j.tcs.2010.05.006}}.

\bibitem{civit2023byzantine}
Pierre Civit, Seth Gilbert, Rachid Guerraoui, Jovan Komatovic, Matteo Monti, and Manuel Vidigueira.
\newblock {Every Bit Counts in Consensus}.
\newblock In {\em 37th International Symposium on Distributed Computing (DISC 2023)}, volume 281, pages 13:1--13:26, Dagstuhl, Germany, 2023.
\newblock \href {https://doi.org/10.4230/LIPIcs.DISC.2023.13} {\path{doi:10.4230/LIPIcs.DISC.2023.13}}.

\bibitem{PODC24:CGGKPV}
Pierre Civit, Seth Gilbert, Rachid Guerraoui, Jovan Komatovic, Anton Paramonov, and Manuel Vidigueira.
\newblock All byzantine agreement problems are expensive.
\newblock In {\em Proceedings of the 43rd ACM Symposium on Principles of Distributed Computing}, PODC '24, page 157–169, New York, NY, USA, 2024. Association for Computing Machinery.
\newblock \href {https://doi.org/10.1145/3662158.3662780} {\path{doi:10.1145/3662158.3662780}}.

\bibitem{PODC23:CGGKV}
Pierre Civit, Seth Gilbert, Rachid Guerraoui, Jovan Komatovic, and Manuel Vidigueira.
\newblock On the validity of consensus.
\newblock In {\em Proceedings of the 2023 ACM Symposium on Principles of Distributed Computing}, pages 332--343, 2023.
\newblock \href {https://doi.org/10.1145/3583668.3594567} {\path{doi:10.1145/3583668.3594567}}.

\bibitem{OPODIS:CGHWW23}
Andrei Constantinescu, Diana Ghinea, Lioba Heimbach, Zilin Wang, and Roger Wattenhofer.
\newblock {A Fair and Resilient Decentralized Clock Network for Transaction Ordering}.
\newblock In {\em 27th International Conference on Principles of Distributed Systems (OPODIS 2023)}, volume 286, pages 8:1--8:20, Dagstuhl, Germany, 2024.
\newblock \href {https://doi.org/10.4230/LIPIcs.OPODIS.2023.8} {\path{doi:10.4230/LIPIcs.OPODIS.2023.8}}.

\bibitem{constantinescu_ghinea_convex_consensus}
Andrei Constantinescu, Diana Ghinea, Roger Wattenhofer, and Floris Westermann.
\newblock {Convex Consensus with Asynchronous Fallback}.
\newblock In Dan Alistarh, editor, {\em 38th International Symposium on Distributed Computing (DISC 2024)}, volume 319, pages 15:1--15:23, Dagstuhl, Germany, 2024.
\newblock \href {https://doi.org/10.4230/LIPIcs.DISC.2024.15} {\path{doi:10.4230/LIPIcs.DISC.2024.15}}.

\bibitem{TCC:DHL21}
Giovanni Deligios, Martin Hirt, and Chen-Da Liu-Zhang.
\newblock Round-efficient byzantine agreement and multi-party computation with asynchronous fallback.
\newblock In {\em Theory of Cryptography Conference}, pages 623--653. Springer, 2021.
\newblock \href {https://doi.org/10.1007/978-3-030-90459-3_21} {\path{doi:10.1007/978-3-030-90459-3_21}}.

\bibitem{EUROCRYPT:Mose24}
Giovanni Deligios and Mose Mizrahi~Erbes.
\newblock Closing the efficiency gap between synchronous and network-agnostic consensus.
\newblock In {\em Advances in Cryptology -- EUROCRYPT 2024}, pages 432--461, 2024.
\newblock \href {https://doi.org/10.1007/978-3-031-58740-5_15} {\path{doi:10.1007/978-3-031-58740-5_15}}.

\bibitem{JACM:DLPSW86}
Danny Dolev, Nancy~A. Lynch, Shlomit~S. Pinter, Eugene~W. Stark, and William~E. Weihl.
\newblock Reaching approximate agreement in the presence of faults.
\newblock {\em J. ACM}, 33(3):499–516, May 1986.
\newblock \href {https://doi.org/10.1145/5925.5931} {\path{doi:10.1145/5925.5931}}.

\bibitem{JACM:DoRe85}
Danny Dolev and R\"{u}diger Reischuk.
\newblock Bounds on information exchange for byzantine agreement.
\newblock {\em J. ACM}, 32(1):191–204, jan 1985.
\newblock \href {https://doi.org/10.1145/2455.214112} {\path{doi:10.1145/2455.214112}}.

\bibitem{DolStr83}
Danny Dolev and H.~Raymond Strong.
\newblock Authenticated algorithms for byzantine agreement.
\newblock {\em SIAM Journal on Computing}, 12(4):656--666, 1983.
\newblock \href {https://doi.org/10.1137/0212045} {\path{doi:10.1137/0212045}}.

\bibitem{JACM:DLS88}
Cynthia Dwork, Nancy Lynch, and Larry Stockmeyer.
\newblock Consensus in the presence of partial synchrony.
\newblock {\em J. ACM}, 35(2):288–323, apr 1988.
\newblock \href {https://doi.org/10.1145/42282.42283} {\path{doi:10.1145/42282.42283}}.

\bibitem{PODC:FisLynMer85}
Michael~J. Fischer, Nancy~A. Lynch, and Michael Merritt.
\newblock Easy impossibility proofs for distributed consensus problems.
\newblock In Michael~A. Malcolm and H.~Raymond Strong, editors, {\em 4th ACM PODC}, pages 59--70. {ACM}, August 1985.
\newblock \href {https://doi.org/10.1145/323596.323602} {\path{doi:10.1145/323596.323602}}.

\bibitem{FLP}
Michael~J Fischer, Nancy~A Lynch, and Michael~S Paterson.
\newblock Impossibility of distributed consensus with one faulty process.
\newblock {\em Journal of the ACM (JACM)}, 32(2):374--382, 1985.
\newblock \href {https://doi.org/10.1145/3149.214121} {\path{doi:10.1145/3149.214121}}.

\bibitem{PODC:GhLiWa22}
Diana Ghinea, Chen-Da Liu-Zhang, and Roger Wattenhofer.
\newblock Optimal synchronous approximate agreement with asynchronous fallback.
\newblock In {\em Proceedings of the 2022 ACM Symposium on Principles of Distributed Computing}, PODC'22, page 70–80, New York, NY, USA, 2022. Association for Computing Machinery.
\newblock \href {https://doi.org/10.1145/3519270.3538442} {\path{doi:10.1145/3519270.3538442}}.

\bibitem{SPAA:GhLiWa23}
Diana Ghinea, Chen-Da Liu-Zhang, and Roger Wattenhofer.
\newblock Multidimensional approximate agreement with asynchronous fallback.
\newblock In {\em Proceedings of the 35th ACM Symposium on Parallelism in Algorithms and Architectures}, SPAA '23, page 141–151, New York, NY, USA, 2023. Association for Computing Machinery.
\newblock \href {https://doi.org/10.1145/3558481.3591105} {\path{doi:10.1145/3558481.3591105}}.

\bibitem{PODC24:GhLiWa}
Diana Ghinea, Chen-Da Liu-Zhang, and Roger Wattenhofer.
\newblock {Brief Announcement: Communication-Optimal Convex Agreement}.
\newblock In {\em {The 43rd ACM Symposium on Principles of Distributed Computing (PODC), Nantes, France}}, June 2024.

\bibitem{PODC25:GhLiWa}
Diana Ghinea, Chen-Da Liu-Zhang, and Roger Wattenhofer.
\newblock Communication-optimal convex agreement.
\newblock In {\em Proceedings of the ACM Symposium on Principles of Distributed Computing}, PODC '25, page 39–49, New York, NY, USA, 2025. Association for Computing Machinery.
\newblock \href {https://doi.org/10.1145/3732772.3733551} {\path{doi:10.1145/3732772.3733551}}.

\bibitem{Gill74}
John~T. Gill.
\newblock Computational complexity of probabilistic turing machines.
\newblock In {\em Proceedings of the Sixth Annual ACM Symposium on Theory of Computing}, STOC '74, page 91–95, 1974.
\newblock \href {https://doi.org/10.1145/800119.803889} {\path{doi:10.1145/800119.803889}}.

\bibitem{GoMoSp23}
Guy Goren, Yoram Moses, and Alexander Spiegelman.
\newblock Probabilistic indistinguishability and the quality of validity in byzantine agreement.
\newblock In {\em Proceedings of the 4th ACM Conference on Advances in Financial Technologies}, AFT '22, page 111–125, New York, NY, USA, 2023. Association for Computing Machinery.
\newblock \href {https://doi.org/10.1145/3558535.3559789} {\path{doi:10.1145/3558535.3559789}}.

\bibitem{WINE18:ByzantineVoting}
Darya Melnyk, Yuyi Wang, and Roger Wattenhofer.
\newblock {Byzantine Preferential Voting}.
\newblock In {\em {14th Conference on Web and Internet Economics (WINE), Oxford, United Kingdom}}, December 2018.
\newblock \href {https://doi.org/10.1007/978-3-030-04612-5_22} {\path{doi:10.1007/978-3-030-04612-5_22}}.

\bibitem{IEEE:MelWat18}
Darya Melnyk and Roger Wattenhofer.
\newblock Byzantine agreement with interval validity.
\newblock In {\em 2018 IEEE 37th Symposium on Reliable Distributed Systems (SRDS)}, pages 251--260, 2018.
\newblock \href {https://doi.org/10.1109/SRDS.2018.00036} {\path{doi:10.1109/SRDS.2018.00036}}.

\bibitem{DIST:MHVG15}
Hammurabi Mendes, Maurice Herlihy, Nitin Vaidya, and Vijay~K Garg.
\newblock Multidimensional agreement in byzantine systems.
\newblock {\em Distributed Computing}, 28(6):423--441, 2015.
\newblock \href {https://doi.org/10.1007/s00446-014-0240-5} {\path{doi:10.1007/s00446-014-0240-5}}.

\bibitem{CCS:AtsRen21}
Atsuki Momose and Ling Ren.
\newblock Multi-threshold byzantine fault tolerance.
\newblock In {\em Proceedings of the 2021 ACM SIGSAC Conference on Computer and Communications Security}, CCS '21, page 1686–1699, New York, NY, USA, 2021. Association for Computing Machinery.
\newblock \href {https://doi.org/10.1145/3460120.3484554} {\path{doi:10.1145/3460120.3484554}}.

\bibitem{Nei94}
Gil Neiger.
\newblock Distributed consensus revisited.
\newblock {\em Information Processing Letters}, 49(4):195--201, 1994.
\newblock \href {https://doi.org/10.1016/0020-0190(94)90011-6} {\path{doi:10.1016/0020-0190(94)90011-6}}.

\bibitem{DISC:NoRy19}
Thomas Nowak and Joel Rybicki.
\newblock {Byzantine Approximate Agreement on Graphs}.
\newblock In {\em 33rd International Symposium on Distributed Computing (DISC 2019)}, volume 146, pages 29:1--29:17, Dagstuhl, Germany, 2019.
\newblock \href {https://doi.org/10.4230/LIPIcs.DISC.2019.29} {\path{doi:10.4230/LIPIcs.DISC.2019.29}}.

\bibitem{OPODIS:StolWat15}
David Stolz and Roger Wattenhofer.
\newblock {Byzantine Agreement with Median Validity}.
\newblock In {\em 19th International Conference on Principles of Distributed Systems (OPODIS 2015)}, volume~46, pages 1--14, Dagstuhl, Germany, 2016.
\newblock \href {https://doi.org/10.4230/LIPIcs.OPODIS.2015.22} {\path{doi:10.4230/LIPIcs.OPODIS.2015.22}}.

\bibitem{PODC:VaiGar13}
Nitin~H. Vaidya and Vijay~K. Garg.
\newblock Byzantine vector consensus in complete graphs.
\newblock In Panagiota Fatourou and Gadi Taubenfeld, editors, {\em 32nd ACM PODC}, pages 65--73. {ACM}, July 2013.
\newblock \href {https://doi.org/10.1145/2484239.2484256} {\path{doi:10.1145/2484239.2484256}}.

\end{thebibliography}

\newpage
\appendix
\section{Appendix}

\subsection{Preliminaries: Missing Proofs}\label{appendix:preliminaries}

\ValidityInclusion*
\begin{proof} Consider a deciding execution $\execution_I$ of $\Pi$ for $I$. Construct the execution $\execution_J$, which is identical to $\execution_I$ except that its input configuration is $J$ instead of $I$. Observe that $\execution_J$ is an execution of $\Pi$ for $J$. Indeed, parties in $\parties(I) \setminus \parties(J)$ (which are byzantine) can act as if they are honest and behave exactly as they do in $\execution_I$. The honest parties in $J$ cannot distinguish between $\execution_I$ and $\execution_J$, so the agreed-upon value in both executions must be the same. Since $\execution_J$ is a valid execution of $\Pi$ for $J$, this value must be in $\validity(J)$.
\end{proof}

\TechnicalEnd*

\begin{proof}
    Note that a countable intersection of events that happen almost surely happens almost surely. We assume the lemma's condition and want to prove that $\validity$ is trivial.
    Because $V_\inputt$ is at most countably infinite, the number of maximal input configurations $I_2 \in \inputconfigs$ is at most countably infinite. By taking an intersection of events, one for each maximal $I_2 \in \inputconfigs$, that happen almost surely by the lemma condition, we get that, almost surely, for all $I_2 \in \inputconfigs$ the canonical executions of $\Pi$ for $I_1$ and $I_2$ decide the same value. Hence, because it happens with probability $1$, it means we can find $\omega \in \Omega$ such that this event happens. That is, for all maximal $I_2 \in \inputconfigs$, the canonical executions of $\Pi$ for $I_1$ and $I_2$ with randomness $\omega$ decide the same value. This implies that, for all maximal $I \in \inputconfigs$, the canonical execution of $\Pi$ on $I$ with randomness $\omega$ decides the same value $v \in V_{\outputt}$. Therefore, by \cref{coro:party-max}, $\validity$ is trivial.
\end{proof}

\subsection{No Cryptographic Setup: Missing Proofs}

\subsubsection{Proof of Lemma \ref{lemma:3-parties}} \label{appendix:ring}
We describe the proof of Lemma \ref{lemma:3-parties}, restated below.
\RingLemma*

To prove our statement, we will be running $A$ in a larger ring containing multiple copies of each party, as depicted in Figure \ref{figure:shinier-ring-appendix}. The ring is constructed from two canonical executions with different input configurations. We will show that parties adjacent in this ring cannot distinguish between the ring and the original setting of three parties, as the third party may be byzantine and simulate the rest of the ring.
This will force parties adjacent in the ring to output the same value, which, in turn, guarantees that the two original executions lead to the same output. 

\paragraph{Restricting the probabilistic space.}
In the following proof, we will consider a countable amount of different executions. Since $\Pi$ satisfies probabilistic termination, it means that each of these executions are deciding almost surely. Because the countable union of almost surely events happens almost surely, this means that the event $E=\{\text{"all executions considered are deciding"}\}$ happens almost surely. So it is enough to prove that all canonical execution of $A$ in $E$ decide the same value to get the proof.



\begin{figure}[t]
\centering
\includegraphics[scale=0.15]{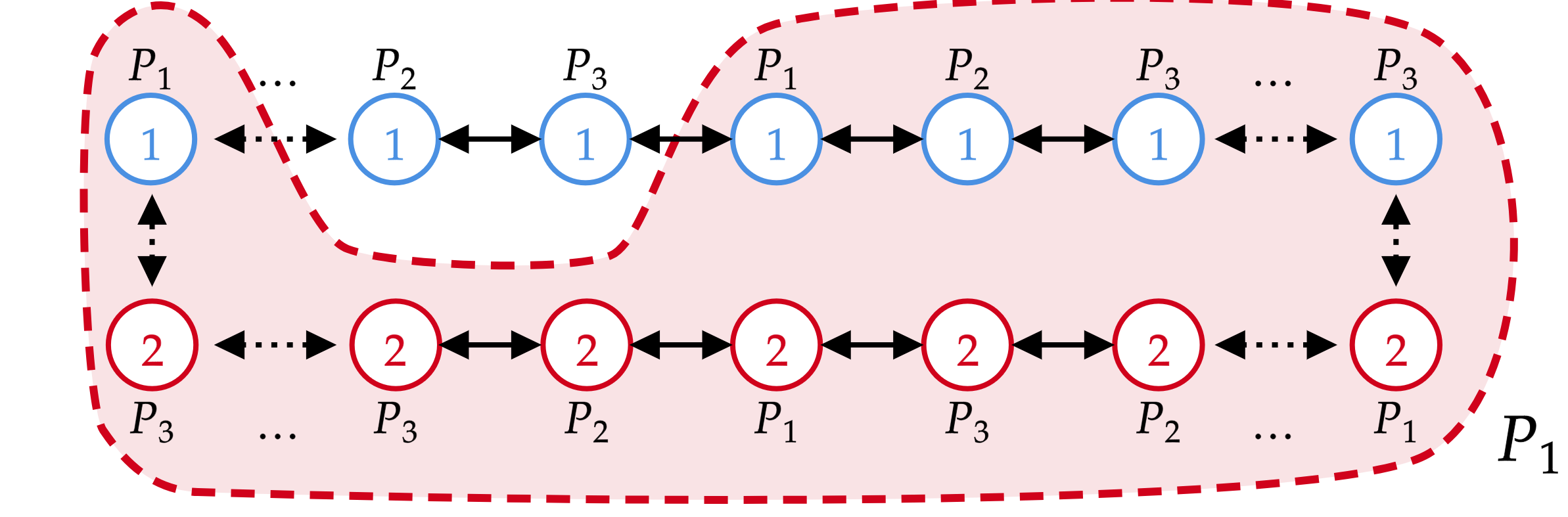}
\caption{Defining the behavior of a byzantine party (here $P_1$).}\label{figure:shinier-ring-appendix}
\end{figure}

\paragraph{Constructing the ring.}
In order to formally describe the construction of the ring, we have to fix two executions with the same randomness. We denote the three parties by $\PS = \{\party_1, \party_2, \party_3\}$, and we consider two arbitrary maximal input configurations $I_1, I_2$. Let $\omega \in E$, and we consider a canonical execution $\varepsilon_1(\omega)$ with input configuration $I_1$, and a canonical execution $\varepsilon_2(\omega)$  with input configuration $I_2$. By definition of $E$, these two executions are deciding.

As $\varepsilon_1(\omega)$ is a deciding execution, there is a number of rounds $r_1(\omega) > 0$ such that, in $\varepsilon_1(\omega)$, all honest parties have decided the same value within $r_1(\omega)$ rounds.
Similarly, there is an $r_2(\omega)$ such that, in the canonical execution $\varepsilon_2(\omega)$, all honest parties have decided the same value within $r_2(\omega)$ rounds. Let $r := \max\{r_1(\omega), r_2(\omega)\}$, $r$ implicitly depends on $\omega$.




To construct the ring depicted in Figure \ref{figure:shinier-ring-appendix}, we make $4(r+1)$ copies of each party $P_i$. Out of the $4(r+1)$ copies of party $\party_i$, $2(r+1)$ will be the copies of $\party_i$ having its input value from $I_1$ and $2(r+1)$ will be the copies of $\party_i$ with input from $I_2$.
The copies of $\party_i$ are then denoted by $\party_{k,i,j}$, where $k \in \{1,2\}$, $i \in \{1,2,3\}$, and $j \in \{0, 1, \ldots, 2r\}$. The copies indexed by $k := 1$ are the ones on the top row of Figure \ref{figure:shinier-ring-appendix}, while the copies indexed by $k := 2$ are the ones on the bottom row.
We now connect these copies via bidirectional communication channels: 
\begin{itemize}
    \item For $k \in \{1,2\}$ and $j \in \{0, 1, \ldots, 2r\}$, we add a channel between $\party_{k,1,j}$ and $\party_{k,2,j}$, and one between  $\party_{k,2,j}$ and $\party_{k,3,j}$.
    \item Then, to complete the path on the row indicated by index $k := 1$, for every index $j \in \{0, 1, \ldots, 2r-1\}$, we add a channel between $\party_{1,3,j}$ and $\party_{1,1,j+1}$.
    \item Similarly, to complete the path on the row indicated by $k := 2$, for each $j \in \{0, 1, \ldots, 2r-1\}$, we add a channel between $\party_{2,1,j}$ and $\party_{2,3,j+1}$.
    \item We now connect the two rows and hence complete the ring: we add a channel between $\party_{1,1,0}$ and $\party_{2,3,0}$, and one between $\party_{1,3,2r}$ and $\party_{2,1,2r}$.
\end{itemize}


\paragraph{Outputs of adjacent copies.} 
We consider a synchronous execution $\varepsilon(\omega)$ on the ring, where each copy $\party_{k, i, j}$  runs $A$ as party $\party_i$, using the input value assigned to it in input configuration $I_k$. We assume that messages are received exactly $\Delta$ units of time after being sent. In the lemma below, we show that adjacent copies in the ring, denoted by $Q_1$ and $Q_2$, obtain the same output.

\begin{lemma} \label{lemma:adjacent-copies}
    Consider two parties $Q_1$ and $Q_2$ that are adjacent in the ring. Then, in execution $\varepsilon(\omega)$, $Q_1$ and $Q_2$ obtain outputs, and they output the same value.
\end{lemma}

\begin{proof}
    Without loss of generality, assume that $Q_1$ and $Q_2$ are copies of $\party_1$ and $\party_2$.
    We consider an execution $\varepsilon'(\omega)$ of $A$ with three parties. In execution $\varepsilon'(\omega)$, $\party_1$ and $\party_2$ have the same input values as $Q_1$ and $Q_2$. $\party_3$ is byzantine and simulates the additional parties in the ring, so they have the same behavior as in execution $\varepsilon(\omega)$.
    All messages are delivered in exactly $\Delta$ time, similarly to execution $\varepsilon(\omega)$.

    We remark that execution $\varepsilon'(\omega)$ and execution $\varepsilon(\omega)$ are identical. Hence, our execution over the ring is equivalent to running $A$ in a synchronous setting with three parties out of which one is corrupted. Therefore, $A$ maintains its properties, namely probabilistic termination and agreement, on the ring as well. It follows that $Q_1$ and $Q_2$ both obtain outputs, and they output the same value.
\end{proof}

\paragraph{Outputs on the entire ring.}
\cref{lemma:adjacent-copies} establishes that adjacent copies in the ring output the same value in execution $\varepsilon'(\omega)$.
Using induction, we prove that all parties output the same value $v$ in $\varepsilon(\omega)$. As a consequence, the copies $\party_{1,i,r}$ and $\party_{2,i,r}$ of each party $\party_i$ output $v$. This will imply that, in the two executions $\varepsilon_1(\omega)$ and $\varepsilon_2(\omega)$ we defined when constructing the ring, all honest parties output the same value.


\begin{lemma}\label{lemma:same-output-ring}
    In execution $\varepsilon(\omega)$, all parties output, and they output the same value.
\end{lemma}
\begin{proof}
We prove by induction that, in execution $\varepsilon(\omega)$, after $r' \leq r$ rounds, 
for $i \in \{1,2,3\}$, parties $\party_{1,i,j}$ with $j \in \{ r - (r-r'), \ldots, r + (r-r') \}$ 
are in the same state as party $\party_i$ in the canonical execution $\varepsilon_1(\omega)$.
The base case $r' := 0$ considers the very beginning of the executions $\varepsilon(\omega)$ and $\varepsilon_1(\omega)$, where the parties have not yet received any messages. Their state then is only defined by their input value and $\omega$. Each party $\party_{1,i,0}$ takes its input from $I_1$, which the case for $\party_i$ in execution $\varepsilon_1(\omega)$ as well. 
We thus obtain that parties $\party_{1,i,0}$ in the ring have the same input state as party $\party_i$.

For the induction step, assume that our claim holds for $r' < r$, and we prove that it also holds for $r' + 1$. Let $\PS_{r'}$ denote the set of parties for which we proved they were in the same state as in $\varepsilon_1(\omega)$ after $r'$ rounds, and let $\PS_{r'+1}$ be the set of parties for which we want to prove that the claim holds after $r'+1$ rounds. The set of direct neighbors of $\PS_{r'+1}$ (including themselves) in the ring is $\PS_{r'}$. Moreover, within one round, parties are only able to receive messages from their direct neighbors. Using the induction hypothesis, we get that parties in $\PS_{r'+1}$ receive exactly the same messages at the $(r'+1)$-th round in execution $\varepsilon_1(\omega)$ and $\varepsilon(\omega)$, therefore they stay in the same state after round $r' + 1$.
    
As a consequence, $\party_i$ and  $\party_{1,i,r}$ are in the same state after the $r$-th round in executions $\varepsilon_1(\omega)$ and $\varepsilon(\omega)$ respectively. However, we have assumed that all parties in $\PS$ output by round $r_1 \leq r$ in execution $\varepsilon_1(\omega)$. Therefore, in execution $\varepsilon(\omega)$, $\party_{1,i,r}$ outputs the same value $v_1$ as  $P_i$ in $\varepsilon_1(\omega)$. With a symmetric argument, one can show that parties $\party_{2,i,r}$ obtain the same output $v_2$ as $P_2$ in $\varepsilon_2(\omega)$. This enables us to conclude that $\varepsilon_1(\omega)$ and $\varepsilon_2(\omega)$ decide the same value  $v_1 = v_2$ using \cref{lemma:adjacent-copies}.
\end{proof}

\paragraph{Assumption on deciding execution.} Before concluding the proof of \cref{lemma:3-parties} using \cref{lemma:same-output-ring}, we need to discuss the assumptions over deciding executions and make sure that event $E$ happens almost surely. For every possible input configuration of size three, we consider a finite amount of fixed executions. Moreover, we take into account that there is only at most a countably infinite amount of input configurations (because we assumed $V_\inputt$ was countable). Therefore, we consider in total a countable union of a finite amount of executions, which is countable. Therefore, the event that all these executions are deciding $E$ happens almost surely, which concludes the proof of \cref{lemma:3-parties}.

\subsubsection{Proof of Theorem \ref{theo:t-sync}} \label{appendix:ring-extended}
We present the proof of Theorem \ref{theo:t-sync}, restated below.
\RingExtended*

\begin{proof}
 Assume $t_s \geq \lfloor n / 3 \rfloor$, and that there is a $t_s$-secure $\ba$ protocol $\Pi$ solving $\validity$ in the synchronous model. 
 
 We consider two arbitrary maximal input configurations $I_1$ and $I_2$.
 We show that all canonical executions of $I_1$ and $I_2$ almost surely decide on the same output value. Then, \cref{lemma:technical-end} ensures that $\validity$ is trivial. Similarly to the proof of \cref{thm:warm-up} and \cref{theo:agnostic-lower-bound}, we use $\Pi$ to build a protocol $A$ for three parties, denoted by $\party_1, \party_2, \party_3$, that matches the setting of \cref{lemma:3-parties}. For protocol $A$, we consider the input space $\{0, 1\}$ and the output space $V_{\outputt}$.




\begin{figure}[t]
\centering
\includegraphics[scale=0.12]{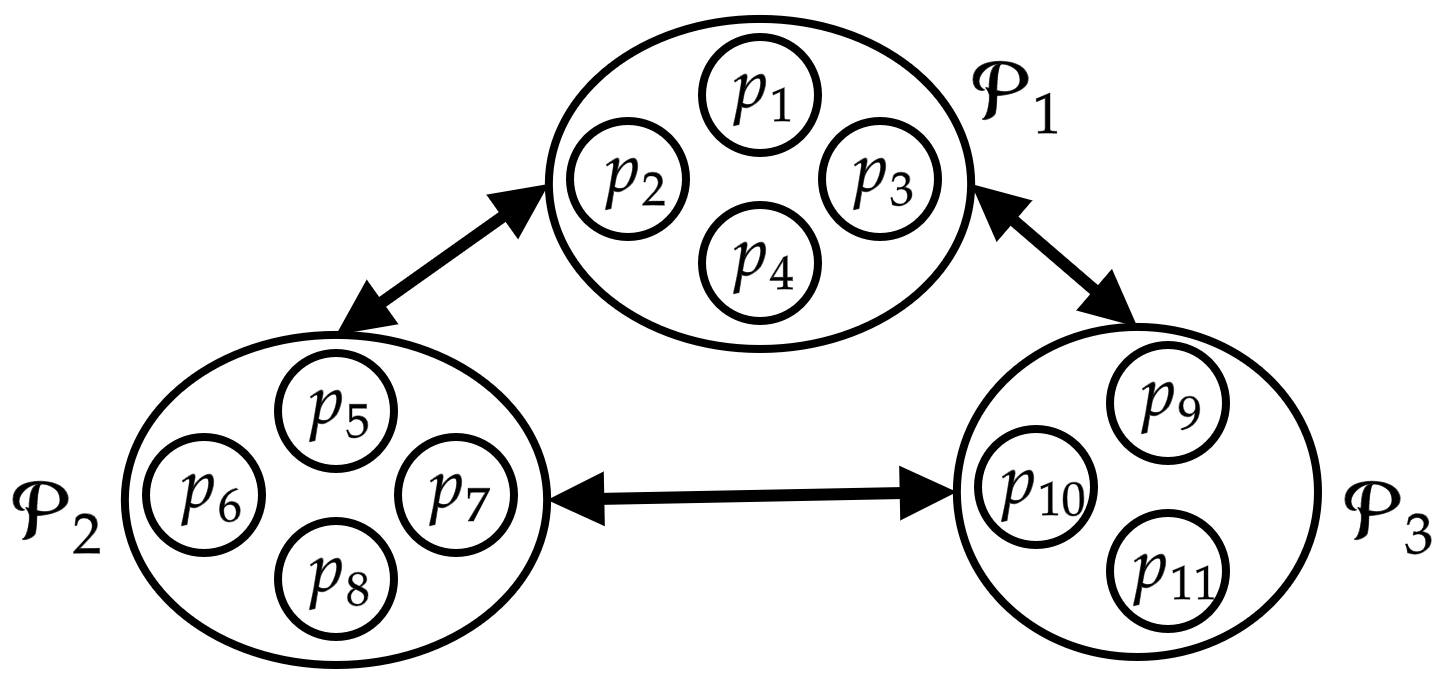}
\caption{Partitioning $\PS$ into $3$ sets $\PS_1, \PS_2, \PS_3$ with $n = 11$}  \label{figure:3-parties-partition}
\end{figure}

To do so, we partition $\PS$ into three sets $\PS_1, \PS_2, \PS_3$ of size at most $t_s$ each. As shown in Figure \ref{figure:3-parties-partition}, in protocol $A$, $\party_i$ simulates all the parties in set $\PS_i$. $\party_i$ ensures that all messages between the simulated parties in $\PS_i$ get delivered within $\Delta$ time. In addition, $\party_i$ forwards any message sent from a party it simulates to a party in set $\PS_j$ ($j \neq i$) to $\party_j$. When $\party_j$ receives this message, it immediately forwards it to the simulated receiver. 
If $\party_i$ has input $0$, the simulated parties in $\PS_i$ take their input from $I_1$. Otherwise, they take their input from $I_2$. When a party in $\PS_i$ outputs a value $v$, $\party_i$ outputs $v$.

Since $\Pi$ is a $t_s$-resilient $\ba$ protocol when $n \leq 3 \cdot t_s$, we obtain that $A$ achieves probabilistic termination and agreement even when one of the three parties is byzantine. Then, \cref{lemma:3-parties} ensures that, almost surely, all deciding canonical executions of $A$ lead to the same output value.

Moreover, we remark that by construction, the canonical execution of $A$ with input values $(0,0,0)$ matches the canonical execution of $I_1$ for $\Pi$. Similarly, the canonical execution of $A$ with input $(1,1,1)$ matches the canonical execution of $I_2$. As a consequence, canonical executions of $I_1$ and $I_2$ decide the same value almost surely. Applying \cref{lemma:technical-end}, we conclude that $\validity$ is trivial.
\end{proof}

\subsection{Similarity Condition: Missing Proof}\label{appendix:similarity}
\SimilarityLemma*
\begin{proof}
    Assume that $\Pi$ is a network-agnostic $\ba$ protocol solving $\validity$. Consider an input configuration $I$.
    
    In the following proof, we will consider a finite amount of different executions. Since $\Pi$ satisfies probabilistic termination, it means that each of these executions are deciding almost surely. 
    Because the finite union of almost surely events happens almost surely, all of these executions will decide almost surely. So, there exists some $\omega \in \Omega$ which we fix such that all the executions below are deciding when ran with randomness $\omega$.
    
    We first consider the canonical execution $\execution_1(\omega)$ of $\Pi$ on $I$.
    As $\Pi$ achieves network-agnostic $\ba$, all honest parties output the same value $v$ in execution $\execution_1(\omega)$. We want to prove that $v \in  \bigcap_{J \in \similar(I)} \validity(J)$, hence we show that $v \in \validity(J)$ for every $J \in \similar(I)$. Using the definition of $\similar(I)$, we split the analysis as follows: 
    \begin{enumerate}[nosep,label=(\roman*)]
        \item $J \in \neighbors(I)$ such that $J \subseteq I$. Consider an execution $\execution_2(\omega)$ where the input configuration is $J$ and the network is synchronous. Parties in $\parties(I) \setminus \parties(J)$ are byzantine, but follow the protocol correctly using the values assigned to them in $I$ as inputs. 
        Parties in $\PS \setminus \parties(I)$ crash at the start of the execution.
        Parties in $\parties(J)$ cannot distinguish between $\execution_2(\omega)$ and $\execution_1(\omega)$, so the output value $v$ must also satisfy $v \in \validity(J)$
        \item $J \in \neighbors(I)$ such that $|J| \geq n - t_a$. First, note that $I \cap J \neq \emptyset$: since 
         $|I| \geq n-t_s$ and $|J| \geq n-t_a$, 
        we obtain that $\abs{I \cap J} \geq n - t_s - t_a$. As $\validity$ is solvable, \cref{theo:agnostic-lower-bound} ensures that $n > 2 \cdot t_s + t_a$, and allows us to conclude that 
        $\abs{I \cap J} > 0$.
        
        Let $\party \in I \cap J$. We consider an execution $\execution_3(\omega)$ where the input configuration is $J$ and the network is asynchronous. Similarly to execution $\execution_2(\omega)$, parties in $\parties(I) \setminus \parties(J)$ are byzantine, but follow the protocol correctly using the values assigned to them in $I$ as inputs, and parties in $\PS \setminus (\parties(I) \cup \parties(J))$ crash at the very beginning of the execution. All messages are delivered in a synchronous way, except for the messages sent from parties in $\parties(J) \setminus \parties(I)$. These are delayed until after party $\party$ outputs: this is possible as
         $\party$ cannot distinguish between $\execution_1(\omega)$ and $\execution_3(\omega)$ and therefore it has to obtain an output without receiving these messages. In addition, the fact that $\party$ cannot distinguish between $\execution_1(\omega)$ and $\execution_3(\omega)$ ensures that the output $v$ agreed upon in $\execution_1(\omega)$ satisfies $v \in \validity(J)$.
    \end{enumerate}

    Therefore, the output $v$ in execution $\varepsilon_1(\omega)$ satisfies $v \in \bigcap_{J \in \similar(I)} \validity(J)$. So when running $\execution_1(\varepsilon)$, if this execution terminates (which happens almost surely), then this value can be used for $\sigma(I)$. 
    
    We now explain how to get a Turing computable function out of protocol $\Pi$ and this property. We first remark that $\Pi$ is a randomized protocol, which can be simulated by a probabilistic Turing machine. However, when time complexity is not taken into account, a probabilistic Turing machine is as expressive as a regular deterministic Turing machine (see Theorem 2 from \cite{Gill74}: for a given number of random bits used, we can try all of them in exponential time). Therefore, we can simulate a deciding execution of $\varepsilon_1$ using a deterministic Turing machine. Because there are no byzantine parties, we can also remove the public-key infrastructure assumption which was used as a black box (for example, one can replace the signing function with one that returns the message concatenated with the id of the party). This in turns gives us a Turing function to compute $\sigma(I)$.
\end{proof}


\end{document}